%% file: main.tex
\documentclass[sigplan,nonacm,10pt]{acmart}

\usepackage{amsmath}
\usepackage{amsthm}
\usepackage{proof}
\usepackage{graphicx}
\usepackage{xspace}
\usepackage{color} 
\usepackage[all]{xy}
\usepackage{tikz-cd}
\usepackage{hyperref}
\usepackage{mathtools}
\usepackage{macros}
\usepackage{MnSymbol}
\usepackage[shortlabels]{enumitem}
\usepackage{libertine}
\usepackage{float}

\usepackage{csquotes}
\usepackage{listings}
\usepackage{macros}
\usepackage{libertinus}
\usepackage{todonotes}

\begin{document}

\title{On the growth rates of polyregular functions}
\author{Mikołaj Bojańczyk (University of Warsaw)}
% \affiliation{University of Warsaw}
\titlenote{This is the author's version of a LICS 2023 paper.}

\begin{abstract}
    We consider polyregular functions, which are certain string-to-string functions that have polynomial output size. We prove that  a polyregular function has output size $\Oo(n^k)$ if and only if it can be defined by an \mso interpretation of dimension $k$, i.e.~a string-to-string transformation where every output position is interpreted, using monadic second-order logic \mso, in some $k$-tuple of input positions. We also show that this characterization does not extend to pebble transducers, another model for describing polyregular functions: we show that for every $k \in \set{1,2,\ldots}$ there is a polyregular function of quadratic output size which needs at least $k$ pebbles to be computed.
\end{abstract}

% \begin{CCSXML}
%     <ccs2012>
%        <concept>
%            <concept_id>10003752.10003766.10003776</concept_id>
%            <concept_desc>Theory of computation~Regular languages</concept_desc>
%            <concept_significance>500</concept_significance>
%            </concept>
%      </ccs2012>
%     \end{CCSXML}
    
    \ccsdesc[500]{Theory of computation~Regular languages}

\keywords{Transductions, monadic second-order logic}
\maketitle

\input{introduction}
\input{qf}

\input{lower}

\input{deatomization-proof}

\bibliographystyle{plain}
\bibliography{bib}

\appendix
\newpage
\input{ap-quant-elim}
\newpage
\input{ap-proof-of-basis}

\clearpage
\input{equivalent-pebble}

\newpage
\input{ap-deatomization}

\end{document}

%% file: introduction.tex
%!TEX root = main.tex

\section{Introduction}
Polyregular functions are a class of string-to-string functions with polynomial growth. Examples of polyregular functions include
\begin{align*}
 \myunderbrace{123 \mapsto 123123}{duplicate} \qquad 
 \myunderbrace{123 \mapsto 123123123}{square}.
\end{align*}
Among many equivalent models defining the polyregular functions, see~\cite{polyregular-survey}, in this paper we work mainly with two models, namely \mso interpretations~\cite[Definition 2]{msoInterpretations} and pebble transducers~\cite[Section 3.1]{DBLP:journals/jcss/MiloSV03}. In an \mso interpretation, the output string is defined using \mso formulas based on the input string, with each position of the output string represented as a $k$-tuple of positions in the input string. A pebble transducer is an extension of a two-way transducer, which instead of a single head, has a stack of at most $k$ pebbles. For both models, it is clear from the definition that if the input has size $n$, then the output size is $\Oo(n^k)$. The purpose of this paper is to investigate if the various numbers $k$ discussed above are really the same number. This corresponds to studying the relationships between the following three hierarchies:

% Other models which are also equivalent include for-transducers and a certain simply typed $\lambda$-calculus, see the survey~\cite{polyregular-survey} for an introduction. 

% For polyregular functions, there are several natural hierarchies that quantify their complexity. %some kind of complexity measure
\begin{itemize}
 \item {\bf The growth rate hierarchy}. A polyregular function is in the $k$-th level of this hierarchy if its growth rate is $\Oo(n^k)$. Here, the \enquote{growth rate} of a string-to-string function is the function that maps an input length $n \in \set{0,1,\ldots}$ to the maximal size of an output that can be produced for inputs of length at most $n$.
% \[ n \in \set{0,1,\ldots} \qquad \mapsto \qquad \max \set{|f(w)| : |w| \le n} \]
% (writing $|w|$ for the length of the string $w$).
% Every polyregular function is on some level, since it has polynomial growth, where intermediate levels such as $\Theta(n^{3.2})$ or $\Theta(n \cdot \log n)$ are impossible. 
% \gaetan{In fact, this result on the impossibility of intermediate
 % levels is already a consequence of [unary output pebbles $\rightarrow$ marbles] and [unary output marbles $\rightarrow$ rational series of polynomial growth]}
% \tito{OK I'm writing a sentence about this after Theorem~\ref{thm:optimization}, is the formulation alright?}
 \item {\bf The dimension hierarchy.} A polyregular function is in the $k$-th level of this hierarchy if it can be defined by an \mso interpretation of dimension $k$, which means that every output position is represented as a tuple of at most $k$ input positions.
 \item {\bf The pebble hierarchy.} A polyregular function is in the $k$-th level of this hierarchy if it can be computed by a pebble transducer that uses a stack of at most $k$ pebbles\footnote{\label{footnote:headcount}Following~\cite{bojanczykPolyregularFunctions2018,lhote2020}, we use the convention that the head of a pebble transducer is counted as a pebble, which means that two-way transducers are one-pebble transducers. Some papers~\cite{engelfriet2002two,DBLP:journals/acta/Engelfriet15,Marble,engelfriet2007xml} do not count the head as a pebble, their $k$-pebble transducers are our $(k+1)$-pebble transducers. The motivation for our choice is that we want the output size of a $k$-pebble transducer to be $\Oo(n^k)$. Also, in our notation there is a meaningful notion of $0$-pebble transducers, which has constant output size $\Oo(1)$; in the alternative notation this notion would need a negative number of pebbles.}.
\end{itemize}
 One can also discuss other hierarchies, e.g.\ the number of nested loops in a for-transducer, but in this paper, we focus on the three hierarchies described above. The natural inclusions between the hierarchies are:
\[
\begin{tikzcd}
\text{$f$ is recognized by a $k$-pebble transducer}
\ar[d,Rightarrow, "\text{\cite[Lemma 2.3]{bojanczykPolyregularFunctions2018}}"]\\
\text{$f$ is defined by an \mso interpretation of dimension $k$}
\ar[d,Rightarrow,"\text{the number of configurations is $\Oo(n^k)$}" ]\\
\text{$f$ is polyregular and has growth rate $\Oo(n^k)$}
\end{tikzcd}
\]
The main results of this paper are:
\begin{itemize}
 \item In Section~\ref{sec:growth-rate-mso-int}, we show that the lower implication is an equivalence, i.e.~the hierarchies for growth rate and dimension are the same, level by level. Furthermore, this hierarchy is computable, i.e.~given a polyregular function one can compute its level $k$ in the growth rate (or equivalently, dimension) hierarchy. In particular, one can check if $k=1$, i.e.~if the function is regular.

 \item In Section~\ref{sec:lower-bound}, we show that the upper implication is not an equivalence, because the pebble hierarchy is slower than the growth rate (or, equivalently, dimension) hierarchy. The hierarchies agree at level $k=1$, but not beyond: for every $k$ there is a polyregular function of quadratic growth which needs at least $k$ pebbles. This result corrects an error in~\cite{lhote2020}, which claimed that all three hierarchies are equal. 

\end{itemize}
% \sandra{I suggest to remove "finite"}

As mentioned above, for level $k=1$ all three hierarchies coincide. The functions on level 1 are a widely studied class of transducers, which can be described by any of the following equivalent models: two-way automata with output~\cite[Note 4]{shepherdson1959reduction}, streaming string transducers~\cite[Section 2.2]{alurCerny11}, string-to-string \mso transductions~\cite[Definition 2]{engelfrietMSODefinableString2001}, regular list functions~\cite[Section 6]{bojanczykRegularFirstOrderList2018}, Church encodings in a linear $\lambda$-calculus~\cite[Theorem~1.2.3]{titoThesis}, etc.\ -- see~\cite{muscholl2019many} for a survey. The two-way automata with output are the same as 1-pebble transducers, while the \mso transductions are the same as \mso interpretations of dimension 1. These models have been shown to be equivalent by Englefriet and Hoogeboom~\cite[Theorem 13]{engelfrietMSODefinableString2001}; this means that for $k=1$ the pebble and dimension hierarchies coincide. Since, as we show in this paper, the dimension and growth rate hierarchies coincide for all $k$, it follows that in the special case of $k=1$, the polyregular functions of linear growth are exactly those that can be defined by \mso transductions, two-way automata with output, and their equivalent models. For this reason, we use the name \emph{linear regular functions} for level $k=1$ of these hierarchies.

Apart from~\cite{lhote2020}, the relationship between the number of pebbles and the growth rate was previously studied for special cases of polyregular functions, namely for comparison-free pebble transducers~\cite[Theorem 7.1]{NguyenNP21} and for marble transducers~\cite[Section~5]{Marble}. In both of these special cases, the pebble hierarchy does coincide with the growth rate hierarchy; unlike the situation for general pebble transducers that we describe in the present paper.

 \paragraph*{Acknowledgement.} This work was financially supported by the Leverhulme Trust, and the Polish National Agency for Academic Exchange. I would also like to thank my colleagues Ga\"etan Doueneau-Tabot, Sandra Kiefer, L{\^{e}} Th{\`{a}}nh D{\~u}ng Nguy{\~{\^{e}}}n and C\'ecilia Pradic for motivating this work, many stimulating discussions, and extensive corrections for drafts of this paper. This paper would not have been possible without their help.

%% file: qf.tex
\section{Growth rate for \mso interpretations}
\label{sec:growth-rate-mso-int}
The section is based on string-to-string \mso interpretations,  one of the equivalent models defining polyregular functions.  In this section, we prove that the growth rate and dimension hierarchies are the same; which implies that a polyregular function has output size $\Oo(n^k)$ if and only if it can be defined by an \mso interpretation that represents output positions using tuples of input positions that have length at most $k$. The proof is based on a detailed analysis of the \mso formulas that are used to define an \mso interpretation. The analysis will be based on the Factorization Forest Theorem of Imre Simon~\cite{simonFactorizationForestsFinite1990}, which we choose to present here as a quantifier-elimination result.

\subsection{\mso interpretations}
\label{sec:mso-int}
We begin by recalling the definition of \mso interpretations and stating the main result.
 We assume that the reader is familiar with basic notions of monadic second-order logic \mso, see~\cite{ebbinghausFlumFinite} for an introduction. We only describe the notation that we use.
A \emph{vocabulary} consists of a finite set of relation names, each one with an associated arity in $\set{0,1,\ldots}$. (So far, only relations are allowed, but later in the paper will also start considering partial functions.) A \emph{structure} over such a vocabulary consists of a finite set, called the \emph{universe} of the structure, and an interpretation of the vocabulary, which associates to each relation name in the vocabulary a relation over the universe of matching arity. The syntax and semantics of first-order logic and \mso are defined in the usual way.
We use the name \emph{class of structures} for a class of such structures over some fixed vocabulary; all classes of structures are assumed to be closed under isomorphism. The structures considered in this paper will be used to describe finite strings and trees; furthermore, the trees will have bounded height. Strings are represented as structures according to the following definition; the representation for trees that we use will be slightly non-standard and will be discussed later on.

\begin{definition}\label{def:ordered-representation}
 For a string $w \in \Sigma^*$, its \emph{ordered representation} is the structure whose universe is the string positions, and which is equipped with the following relations 
 \begin{align*}
 \myunderbrace{x \le y}{order on \\ positions} 
 \qquad \qquad \myunderbrace{a(x)}{$x$ has label $a$,\\ for every $a \in \Sigma$}.
 \end{align*}
\end{definition}

An alternative representation would be the \emph{successor representation}, in which the order is replaced by the successor relation. This representation is equally good when defining languages in \mso, but it leads to problems when defining functions, as explained in~\cite[Theorem 4]{msoInterpretations}.

We use \mso to define functions, and not just languages; these functions are called \mso interpretations. The idea is that an \mso interpretation uses $k$-tuples of input elements, for some fixed dimension $k \in \set{0,1,\ldots}$, to describe output elements. Which $k$-tuples participate in the output, and how the relations of the output structure are defined on the output elements -- all of this is described using \mso formulas. The formal definition is given below. It allows extra features of minor importance, namely, we can use several copies of the input structure (called \emph{components}), and each copy can use a different dimension. These extra features are used to make the definition more robust, e.g.~so that functions with constant output size can be defined using dimension $k=0$.

\begin{definition}
 \label{def:mso-interpretation}
 An \mso interpretation is a function 
 \begin{align*}
 f : \Cc \to \Dd
 \end{align*}
 between two classes of structures that is defined as follows. All formulas below are \mso formulas over the vocabulary of the input class $\Cc$. All free variables in the formulas have element type, but the formulas are allowed to quantify over sets.
 \begin{enumerate}
 \item {\bf Components.} There is a finite set $Q$, whose elements are called \emph{components} of the interpretation. Each components has an associated \emph{dimension} in $\set{0,1,\ldots}$. 
 \item {\bf Universe formulas. } For every component, there is an associated \emph{universe formula}, whose number of free variables is equal to the \emph{dimension} of the component. The universe formulas define the universe of the output structure in the following way: if the input structure is $A \in \Cc$ then the universe of the output structure is the disjoint union
 \begin{align*}
 \coprod_{q \in Q}\setbuild{ \bar a \in A^{\text{dimension of $q$}}}{$\bar a$ satisfies the universe \\ formula for $q$}.
 \end{align*}

 \item {\bf Relation interpretations. } For every relation name $R$ in the vocabulary of the output class $\Dd$, say of arity $\ell$, and for every components $q_1,\ldots,q_\ell \in Q$, there is a formula $\varphi$ 
 such that for every input structure $A \in \Cc$,
 \begin{align*}
 A \models \varphi(\bar a_1 \cdots \bar a_\ell) \quad \text{where }\bar a_i \in A^{\text{dimension of }q_i}
 \end{align*}
holds if and only if the relation $R$ in the output structure selects the $\ell$-tuple in which the $i$-th coordinate is $\bar a_i$ from component $q_i$. 
% \item {\bf Function interpretations.} Functions in the output structure are described in the same way as relations, with a partial function of arity $\ell$ viewed as a relation of arity $\ell+1$.
 \end{enumerate}
\end{definition}

The above definition generalizes languages. A language can be seen as a function 
\begin{align*}
L : \Cc \to 2
\end{align*}
where $2$ is some class of structures that contains two structures, representing  ``true'' and ``false''. If the two structures in the output class have at most $n$ elements, then $n$ components of dimension $0$ can be used. 

For general structures, such as graphs, \mso interpretations are not particularly well behaved, in particular, they are not closed under function composition, see the comments in~\cite[Exercise 11.2.4]{ebbinghausFlumFinite} or~\cite[Theorem 4]{msoInterpretations}. However, good behaviour is recovered in the string-to-string case. Define a  \emph{string-to-string \mso interpretation} to be an \mso interpretation of type $\Sigma^* \to \Gamma^*$, where the input and output alphabets are finite and  strings are modelled as structures using  the ordered representation from Definition~\ref{def:ordered-representation}.  Such interpretations define exactly the \emph{polyregular functions}~\cite[Theorem 7]{msoInterpretations}; the latter being a class of string-to-string functions described in~\cite{polyregular-survey}. As far as this paper is concerned, we can view string-to-string \mso interpretations as the definition of the class of polyregular functions. Since polyregular functions are closed under composition~\cite[Theorem 1.4]{polyregular-survey}, the same is true for string-to-string \mso interpretations.

\paragraph*{The growth rate of string-to-string functions.} 
We now state the main result of this section, which describes the growth rate of string-to-string \mso interpretations.
The \emph{dimension} of an \mso interpretation is defined to be the maximal dimension of its components. The dimension gives a simple upper bound on the growth rate: if an \mso interpretation has dimension $k$ and the input structure has $n$ elements, then the output structure will clearly have $\Oo(n^k)$ elements. The main result of this section is that for string-to-string functions, one can choose the \mso interpretation so that this simple bound is tight. 

\begin{theorem}\label{thm:optimization} For every polyregular function one can compute some $k \in \set{0,1,\ldots}$ such that the function is an \mso interpretation of dimension $k$ and has growth rate $\Theta(n^k)$. 
 \end{theorem}

The theorem is proved by taking an \mso interpretation, and eliminating redundant variables in the universe formulas, so that the remaining variables are independent enough to make the dimension optimal. This is done in two steps. The first step, in Section~\ref{sec:quantifier-elimination}, is a quantifier elimination result for polyregular functions. The second step, in Section~\ref{sec:proof-of-optimization-theorem}, uses the quantifier elimination result to prove Theorem~\ref{thm:optimization}.

\subsection{Quantifier elimination}
\label{sec:quantifier-elimination}
In this section we show that every polyregular function can be decomposed into two stages: the first stage is a linear preprocessing of the input, and the second stage is a quantifier-free interpretation, i.e.~an \mso interpretation where all formulas are quantifier-free. The intermediate structure produced in the first stage is not a string, but a tree of bounded height.

\paragraph*{Function symbols.} When eliminating quantifiers, we use structures that have not only relations, but also function symbols, which are interpreted as partial functions. Formally speaking, the vocabularies can also have function symbols, also with associated arities in $\set{0,1,\ldots}$. In a structure with universe $A$, a function symbol of arity $\ell$ is interpreted as a partial function of type $A^\ell \to A$. When defining the semantics of first-order logic in the presence of partial functions, we assume that an atomic formula holds if all of the partial functions used in it have defined values, and the corresponding relation is satisfied. For example, if a structure has a constant $c$ (a constant is a function of arity zero) which is undefined, then any atomic formula involving this constant, such as $R(c,c)$ or $c=c$, will be false.

\paragraph*{Trees.} The partial function symbols will be used to describe operations on tree nodes, such as the parent operation. We use trees which are node labeled and sibling ordered. The following picture explains our tree terminology: 
\mypic{20}
Although the sibling successor is a partial function, we view it as a relation, since otherwise quantifier-free formulas could iterate the sibling successor to look arbitrarily far in the tree.
We use the name \emph{sibling order} for the reflexive transitive closure of the  sibling successor relation; this is a union of total orders, one for each set of siblings. 

We intend to use trees as a representation of the trees that appear in the Factorization Forest Theorem of Imre Simon~\cite[Theorem 1]{bojanczykFactorizationForests2009}, in particular trees will have bounded height. (The \emph{height} of a tree is defined to be the maximal number of edges on a root-to-leaf path.) That is why the structure in the following definition is named after Simon. The structure is chosen so that all relevant information from the point of view of this theorem can be accessed in a quantifier-free way. 

\begin{definition}[Simon representation of a tree]\label{def:simon-representation}
 For a tree with nodes labeled an alphabet $\Sigma$, its \emph{Simon structure} is the structure in which the universe is the tree nodes, and which is equipped with the following functions and relations:
\begin{enumerate}
    \item a constant for the root;
    \item a unary function that maps each node to its parent, and which is undefined for the root node;
    \item a unary relation selecting nodes with label $a \in \Sigma$;
    \item a unary relation selecting leftmost siblings;
    \item a unary relation selecting  rightmost siblings;
    \item a binary relation for the sibling order;
    \item a binary relation for the sibling successor.
\end{enumerate}
\end{definition}

For example, the quantifier-free formula 
\begin{align*}
 \text{sibling-successor}( \text{parent}(x),\text{parent}(y))
\end{align*}
says that the parents of nodes $x$ and $y$ are sibling successors. In particular, $x$ and $y$ have the same grandparent.

\paragraph*{Tree grammars.} As mentioned before, we intend to work with bounded height trees. To make the height bounded, we will generate trees using certain grammars, which are called \emph{tree grammars} in this paper, and which syntactically ensure bounded height. A tree grammar consists of a finite set of labels $\Sigma$ with a distinguished \emph{root label}, and a set of rules, with each rule having one of four kinds: 
\begin{align*}
\myunderbrace{a \to }{nullary rule}
\qquad
\myunderbrace{a \to b}{unary rule}
\qquad
\myunderbrace{a \to bc}{binary rule}
\qquad
\myunderbrace{a \to b^*}{star rule}
\end{align*}
The rules are required to be \emph{acyclic}, which means that there is a pre-order on the letters such that for every rule, the letter before the arrow is strictly bigger than all letters after the arrow.
The semantics of a grammar is a set of ordered trees with nodes labeled by $\Sigma$, which is defined in the natural way and explained in the following picture:
\mypic{16}
Acyclicity ensures that the height of trees generated by the grammar is bounded by the size of the alphabet. The set of trees generated by a tree grammar is viewed as a class of structures using the Simon representation. 

\paragraph*{Quantifier elimination for polyregular functions.} We now present the quantifier elimination result for polyregular functions. This theorem can be seen as an abstraction of the Factorization Forest Theorem, which encapsulates the properties of factorization trees that are needed in our context. We believe that this perspective, which views the Factorization Forest Theorem as a quantifier elimination result, might be useful in future work\footnote{This perspective is not entirely new. Already in~\cite[Lemma 1]{colcombetCombinatorialTheoremTrees2007}, Colcombet views factorization trees as a data structure that allows one to reduce \mso to first-order logic. Kazana and Segoufin take this one step further in~\cite[Theorem 3.2]{segoufinKazan2013}, by observing that the reduction yields special formulas of first-order logic, namely those with quantifier prefix $\exists^* \forall^*$. Here, we take these observations one step further, by putting enough structure in the tree so that the formulas become quantifier-free.}. 

The idea behind the quantifier elimination result, stated in Theorem~\ref{thm:quantifier-elimination} below, is that each input string can be equipped with a tree structure of bounded height, such that a given polyregular function can be computed in a quantifier-free way based on this structure. 
 In the theorem, the \emph{yield} of a tree is defined to be the string consisting of the labels of leaves in the tree, read from left to right. 

\begin{theorem}\label{thm:quantifier-elimination}
For every polyregular string-to-string function 
\begin{align*}
f : \Sigma^* \to \Gamma^*
\end{align*}
there is a tree grammar $\Tt$, an \mso interpretation $h$ that is linear (i.e.~dimension at most one) and a quantifier-free interpretation $g$ such that the following diagram commutes: 
\[
 \begin{tikzcd}
 [column sep=1.5cm]
 \Sigma^* 
 \ar[r,"\text{linear }h"]
 \ar[dr,equals, bend right=30]
 & \Tt 
 \ar[dr,"\text{quantifier-free }g"]
 \ar[d, "\text{yield}"']
 \\
 & 
 \Sigma^*
 \ar[r,"f"']
 &
 \Gamma^*
 \end{tikzcd}
 \] 
\end{theorem}

The theorem is proved in the appendix, using the Factorization Forest Theorem. 
We only explain here how the interpretations above  handle partial functions; this was not explained in Definition~\ref{def:mso-interpretation} which used only vocabularies without function names.  Recall that strings are represented using the ordered representation from Definition~\ref{def:ordered-representation} and trees are represented using  the Simon representation from Definition~\ref{def:simon-representation}. For the linear interpretation $h$, which inputs strings and outputs trees, the functions in the output tree are viewed as relations that represent their graphs, i.e.~a function with $\ell$ arguments is viewed as a relation with $\ell+1$ arguments\footnote{This view would be overly simplistic for interpretations that are both quantifier-free and output structures with functions. However, in our setting, the interpretations are either quantifier-free or output structures with functions, but not both. }. In this particular situation, $\ell \le 1$ because the functions in Definition~\ref{def:simon-representation} have at most one argument. For the quantifier-free interpretation $g$, the only partial functions are in the input class, so there is no need to adapt Definition~\ref{def:mso-interpretation}, other than allowing the formulas to use the functions from the input structure.

\subsection{Proof of Theorem~\ref{thm:optimization}}
\label{sec:proof-of-optimization-theorem}
In this section, we use the quantifier-elimination result from Theorem~\ref{thm:quantifier-elimination} to complete the proof of Theorem~\ref{thm:optimization} about the growth rate of polyregular functions. 

Take a polyregular function $f$. We want to show that there is some $k \in \set{0,1,\ldots}$ such that this function has growth rate $\Theta(n^k)$ and can be defined by an \mso interpretation of dimension $k$. 
Apply Theorem~\ref{thm:quantifier-elimination} to the function $f$, yielding
\[
 \begin{tikzcd}
 [column sep=1.5cm]
 \Sigma^* 
 \ar[r,"\text{linear }h"]
 \ar[dr,equals, bend right=30]
 & \Tt 
 \ar[dr,"\text{quantifier-free }g"]
 \ar[d, "\text{yield}"']
 \\
 & 
 \Sigma^*
 \ar[r,"f"']
 &
 \Gamma^*
 \end{tikzcd}
 \] 

We will show that the growth rate and dimension coincide for quantifier-free interpretations, as explained in the following lemma. The following lemma ashows that for quantifier-free interpretations, the growth rate can be computed, and corresponds to the optimal dimension in some first-order interpretation.

\begin{lemma}\label{lem:optimization-for-quantifier-free}
 For every quantifier-free interpretation $g : \Tt \to \Gamma^*$
 one can compute some $k \in \set{0,1,\ldots}$ such that $g$ is a first-order interpretation of dimension $k$ and has growth rate $\Theta(n^k)$.
 \end{lemma}

From the lemma, we immediately get the same result for \mso interpretations. Apply the lemma to the quantifier-free interpretation $g$, yielding some $k$. The yield operation is length-preserving if we define the length of a tree to be the number of leaves. Since the yield operation from $\Tt$ is also surjective, it follows that the growth rates are the same for $f$ and $g$, namely $\Theta(n^k)$. Also, $f$ is an \mso interpretation of dimension $k$, as a composition of a linear  \mso interpretation with a first-order interpretation of dimension $k$. Thus, we have proved that $f$ has growth rate $\Theta(n^k)$ and is an \mso interpretation of dimension $k$, completing the proof of Theorem~\ref{thm:optimization}. We are left with Lemma~\ref{lem:optimization-for-quantifier-free}.

The rest of this section is devoted to proving Lemma~\ref{lem:optimization-for-quantifier-free}. This will be done via a syntactic analysis of quantifier-free types. Here, a \emph{quantifier-free type} is defined to be a quantifier-free formula $\varphi(x_1,\ldots,x_\ell)$ such that every quantifier-free formula with the same free variables is either implied by $\varphi$ or inconsistent with it. In this paper, we care about quantifier-free types that arise by taking some tree in a tree grammar, and describing the quantifier-free formulas that are satisfied by some tuple of $\ell$ distinguished nodes. Such a quantifier-free type will describe the distinguished nodes and their ancestors, using the relations available in the Simon representation. 

\begin{myexample}\label{ex:dependencies-in-skeleton}
    The following picture shows a quantifier-free type which arises by taking a tree with six distinguished nodes, using the Simon representation.  In the picture, we use ellpises $\ldots$ to represent missing nodes. The presence or absence of missing nodes can be deduced from the relations for leftmost siblings, rightmost siblings, and successor siblings, which are available in the Simon representation. 
\mypic{34}
There are certain functional dependencies between the distinguished nodes in the above type. Here are some:
\begin{center}
    \begin{tabular}{ll}
        dependency & reason \\ \hline
        $x_1$ determines $x_5$ & successor sibling\\
        $x_1$ determines $x_6$ & rightmost sibling\\
        $x_1$ determines $x_2$ & successor sibling of grandparent\\
        $x_3$ determines $x_4$ & leftmost child
    \end{tabular}
\end{center}

The above list is non-exhaustive, for example the dependency between $x_3$ and $x_4$ is mutual. Thanks to the dependencies described above, the distinguished nodes $x_1$ and $x_3$ determine all the other distinguished nodes. Since these two nodes do not determine each other, they are what we call a \emph{basis} of the distinguished nodes. The basis is not unique, e.g.~we can replace $x_3$ with $x_4$. As we will see later, the size of the basis is unique. Since the basis has size two, and the nodes in it can be chosen indepdently, the growth rate of the quantifier-free type in this example is quadratic, i.e.~the number of realizations in a given tree is at most quadratic, and there are trees in which it is at least quadratic.
\end{myexample}

The analysis from the above example is formalized in the following lemma. 

 \begin{lemma}[Basis lemma]\label{lem:seed-qf-formulas} Let $\Tt$ be a tree grammar, and let  $\varphi(x_1,\ldots,x_\ell)$ be a quantifier-free type over the vocabulary of $\Tt$, using Simon representation. There is a subset 
 \begin{align*}
 X \subseteq \set{x_1,\ldots,x_\ell}
 \end{align*}
 of the free variables which is a basis in the following sense:
 \begin{enumerate}
 \item \label{basis:spans} The variables in $X$ span all variables, in the sense that for every tree $t \in \Tt$, if two tuples selected by $\varphi$ agree on the variables from $X$, then they are equal.
 \item \label{basis:growth} The variables in $X$ are independent, in the sense that the following function is in $\Theta(n^k)$, where $k = |X|$:
\begin{align*}
\qquad \qquad 
\myunderbrace{n \in \set{1,2,\ldots} 
\ \mapsto \ 
\begin{tabular}{c}
 \text{maximal number of tuples } \\
 \text{ that can be selected by $\varphi$ in  }\\
 \text{a  tree from $\Tt$ of size at most $n$}
\end{tabular}}
{this function is called the \emph{growth rate} of $\varphi$}
\end{align*}
 \end{enumerate} 
\end{lemma} 

The Basis Lemma is shown in the appendix, using a syntactic analysis of dependencies between variables in a quantifier-free type. We now show how it implies Lemma~\ref{lem:optimization-for-quantifier-free}, about growth rate and dimension coinciding for quantifier-free interpretations, and thus also Theorem~\ref{thm:optimization}.

\begin{proof}[Proof of Lemma~\ref{lem:optimization-for-quantifier-free} using the Basis Lemma] Let  \begin{align*}
 g : \Tt \to \Gamma^*
 \end{align*}
 be a quantifier-free interpretation
 as in the assumption of Lemma~\ref{lem:optimization-for-quantifier-free}. For each component, consider its universe formula. Since the trees in the tree grammar $\Tt$ have bounded height, there are finitely many possible quantifier-free types for a given number of variables, and each quantifier-free formula is equivalent to a disjunction of some quantifier-free types.  Therefore, by possibly increasing the number of components,  we can assume without loss of generality that for every component, the corresponding universe formula is a quantifier-free type. For each component, apply the Basis Lemma for the corresponding quantifier-free type, yielding some basis. We will use the first item of the Basis Lemma to reduce the dimension of each component to the size of the basis, and the second item to give a matching lower bound for the growth rate. 
 
 Consider one of the components, and a basis for the universe formula, which is a subset of its free variables. By the first item of the Basis Lemma, all other variables in the universe formula are spanned by the basis variables. Therefore, we can reduce the dimension of this component to the size of the basis, as follows. The new universe formula uses only the basis as its free variables, and it holds if the new free variables can be extended to some tuple that satisfies the original universe formula. (The extension, if it exists, is unique by item~\ref{basis:spans} of the Basis Lemma.) Observe that the new universe formula is no longer quantifier-free, because it uses existential quantifiers; nevertheless, it is a first-order formula (even an existential one) and not an \mso formula, since no sets need to be quantified. The remaining formulas in the interpretation, which describe the relations of the output string, are adjusted accordingly, by applying the original quantifier-free formulas to the unique extensions.
 
 We now use the second item in the Basis Lemma to argue that the new interpretation has optimal dimension. 
 Let $k$ be the maximal size of the bases used in the construction above. By the second item of the Basis Lemma, we know that the growth rate of one of the universe formulas is $\Theta(n^k)$ (this is true for both the original and new interpretations), and therefore the growth rate of the function $g$ is $\Theta(n^k)$, which matches the dimension of the new interpretation.
\end{proof}

%% file: lower.tex
\section{On the cost of stack discipline}
\label{sec:lower-bound}

In Theorem~\ref{thm:optimization} we showed that the growth rate and dimension hierarchies coincide for polyregular functions.  In this section, we show that the correspondence fails for the hierarchy which counts the number of pebbles in a  pebble transducer, and it fails badly: there is no level of the pebble hierarchy that covers all quadratic polyregular functions.  This result and its proof correct an error in~\cite[Theorem 10]{lhote2020}, where it was claimed that the pebble hierarchy coincides with the growth rate hierarchy.

\paragraph*{Pebble transducers.} The usual definition of pebble transducers, see~\cite[Section 1]{engelfriet2002two} is operational, and it describes an extension of a two-way automaton with pebbles used to mark positions in the input. In this paper, we use a slightly non-standard approach to pebble transducers -- we define them as a special case of string-to-string \mso interpretations. This is done by using automata terminology (such as state and configuration) for an \mso interpretation, and then imposing a restriction called stack discipline.

We begin by describing the automata terminology for  string-to-string \mso interpretations. Instead of component, we use the name \emph{state}. Define a \emph{configuration} of an \mso interpretation to be a tuple that consists of an input string, a state, and a list of positions  that satisfies the universe formula for the state. The list of positions is called the \emph{pebble stack} of the configuration; the length of this list, which is the dimension of the corresponding state, is called the \emph{stack height}. The \emph{head} of the configuration is defined to be the last position in the pebble stack. Here is a picture of a configuration: 
\mypic{21}
We say that two configurations are \emph{consecutive} if they have the same input string, and they are consecutive elements according to the linear order on configurations given by the transducer. 
% A pebble transducer is defined to be the special case of a string-to-string \mso interpretation in which the pebble stacks are updated in a way that respects stack discipline. 

% A \emph{run} is defined to be any list of configurations that is obtained by taking some two configurations, called the \emph{source} and \emph{target}, which share the same input string, and listing all configurations between them (including the source and target), according to the linear order on configurations that is given in the interpretation. An \emph{accepting run}  is a run that is maximal, i.e.~its source configuration has no predecessor and its target configuration has no successor. Using the above defined terminology for \mso interpretations, we can define pebble transducers as the special case which satisfies the stack discipline condition described below. 
\begin{definition}\label{def:pebble-transducer}
  A \emph{pebble transducer} is a string-to-string \mso interpretation that satisfies the following \emph{stack discipline condition}. For every two consecutive configurations, either 
  \begin{enumerate}
    \item {\bf Push/pop.} One of the two pebble stacks is a prefix of the other; or 
    \item {\bf Move.} Both pebble stacks have equal lengths, and are equal except for the head.
  \end{enumerate}
\end{definition}

%  In a string-to-string \mso interpretation as described in Definition~\ref{def:mso-interpretation}, each configuration contributes exactly one letter to the output string. When dealing with pebble transducers, it will be more convenient to use a slightly generalized output mechanism: for each state, there is an associated string over the output alphabet, possibly empty; this string is used as the output for each configuration with that state. It is easy to see that this generalization does not change the expressive power of the model.  From now on, all pebble transducers will use this generalized output mechanism.

 When speaking of pebble transducers, the dimension is called the \emph{number of pebbles}. A $k$-pebble transducer is one with $k$ pebbles, i.e.~the maximal stack size is $k$. 
 
 We now show that  our definition of pebble transducers is equivalent to the one usually found in the literature. There is a small proviso: in order to consistently compare our model with the one in the literature, we need to count pebbles in the same way for both models in the same way, we  count the head as a pebble, see  Footnote~\ref{footnote:headcount}.

 \begin{lemma} \label{lem:equivalent-pebbles} For every number of pebbles $k \in \set{1,2,\ldots}$, the model from Definition~\ref{def:pebble-transducer} computes the same string-to-string functions as the model  defined in~\cite[Section 1]{engelfriet2002two}.  
 \end{lemma}

%  The proof of the above lemma, together with the definition of the model from~\cite{engelfriet2002two}, is in the appendix. We would like to remark that in order to  consistently compare the two models in the above lemma, we count the head as a pebble in the model from~\cite{engelfriet2002two}, which is different from the way that the authors of~\cite{engelfriet2002two} count pebbles. 

 From now on, when talking about pebble transducers, we use the model from Definition~\ref{def:pebble-transducer}.

\paragraph*{The pebble hierarchy does not coincide with growth rates.} Pebble transducers compute the same string-to-string functions as \mso interpretations, see
~\cite[Theorem 7]{msoInterpretations}. However, the construction of a pebble transducer from an \mso interpretation in~\cite{msoInterpretations} increases the dimension. In this section, we prove that the tradeoff is indeed necessary: already the quadratic growth polyregular functions cannot be captured by any finite level of the pebble hierarchy (the hierarchy of polyregular functions that is indexed by the number of pebbles needed to compute a function). 

\begin{theorem}\label{thm:lower-bound}
    For every $k$ there is a  polyregular function that has quadratic growth rate and which is not recognized by any $k$-pebble transducer.
\end{theorem}

Since we have already proved that quadratic growth rate is the same as being defined by an \mso interpretation of dimension two, an alternative phrasing of the above theorem is that    \mso interpretations of dimension $k=2$ define strictly more functions than two pebble transducers, or three pebble transducers, etc. In other words, imposing the stack discipline on an \mso interpretation might result in an arbitrary increase in its dimension.

Before proving the lower bound from Theorem~\ref{thm:lower-bound}, we observe that there are no problems\footnote{However, in the case of \emph{for-transducers} (one of the equivalent models defining polyregular functions~\cite[Section~3]{bojanczykPolyregularFunctions2018}), a similar phenomenon appears already for functions of linear size increase: for every $k \in \set{1,2,\ldots}$ there is a linear regular function that requires at least $k$ nested loops in a for program that recognizes it. We do not describe this example in detail; the idea is to nest the reverse operation $k$ times.
} for  functions of linear growth. This is because in the case of $k=1$, stack discipline is a vacuous condition, and therefore one pebble transducers compute exactly the same function as \mso interpretations of dimension one.

This section is devoted to proving Theorem~\ref{thm:lower-bound}. We begin by illustrating the proof strategy with a function that has quadratic growth, and yet nevertheless requires three pebbles to be computed\footnote{A variant of this  function was first suggested by L{\^{e}} Th{\`{a}}nh Dung Nguy{\^{e}}n and Ga\"etan Dou\'eneau-Tabot.}.  
This function, which will be called \emph{block squaring},  inputs a sequence of $n$ blocks of $a$ letters delimited by brackets and outputs each pair of blocks:
\begin{align*}
\langle a^{k_1} \rangle 
\cdots 
\langle a^{k_n} \rangle 
\quad \mapsto \quad 
\langle a^{k_1} | a^{k_1} \rangle 
\cdots 
\langle a^{k_i} | a^{k_j} \rangle 
\cdots 
\langle a^{k_n} | a^{k_n} \rangle 
\end{align*}
The pairs of blocks in the output string are ordered lexicographically, as in the following example
\begin{align*}
  \langle a^{1} \rangle 
  \langle a^{2} \rangle 
  \langle a^{3} \rangle 
  \ \mapsto \ &
  \langle a^{1} |  a^{1} \rangle 
  \langle a^{1} |  a^{2} \rangle 
  \langle a^{1} |  a^{3} \rangle 
  \langle a^{2} |  a^{1} \rangle 
  \langle a^{2} |  a^{2} \rangle \\ &
  \langle a^{2} |  a^{3} \rangle 
  \langle a^{3} |  a^{1} \rangle 
  \langle a^{3} |  a^{2} \rangle 
  \langle a^{3} |  a^{3} \rangle 
\end{align*}
 If the input is ill-formatted, i.e.~it does not belong to the regular language $(\langle a^* \rangle)^*$, then the output is empty. 
The growth rate of this function is easily seen to be quadratic. We can compute the function using three pebbles as follows: if there are $n$ blocks in the input, then the first two pebbles range over pairs $(i,j)$ blocks, ordered lexicographically. The lexicographic order is consistent with stack discipline, with coordinate $i$ corresponding to the bottom of the stack. Once we have selected such a pair, we need to output the $i$-th block and the $j$-th block. Since the pebble pointing to block $j$ is at the top of the stack, there is no need for extra pebbles to output the $j$-th block. However, to copy the $i$-th block without losing the pebble that points to the $j$-th block, we need an extra third pebble. 

It remains to show the lower bound for block squaring, i.e.~that it cannot be computed by a two pebble transducer. The intuitive reason was described in the previous paragraph; when we want to copy a block from the input to the output, the head of the pebble transducer should be pointing to that block.  However, this idea is not exactly correct -- for example, a pebble transducer could first check if all input blocks have length exactly two, and for such inputs, it could use a specially crafted procedure that takes advantage of this knowledge. Our lower bound proof needs to take into account such pebble transducers. 

Because of such difficulties,  in Section~\ref{sec:atom-lower-bound} we begin by studying an abstraction of the function described above, which uses elements from an infinite $\atoms$ to represent blocks of the form $\langle a^n \rangle$.  The elements of this set will be called \emph{atoms}, and we will use a transducer model which is not allowed to inspect the atoms in any way, and can output atoms only by indicating an atom with its head. The corresponding abstraction of the block squaring function  is the function 
\begin{align*}
\myunderbrace{a_1 \cdots a_n \in \atoms^* \qquad \mapsto  \qquad a_1 a_1 \cdots a_i a_j \cdots a_n a_n,}{call this function \emph{atom squaring}}
\end{align*}
in which the atom pairs are ordered lexicographically.
This function is quadratic, but we will show that it needs at least three pebbles, under a suitable adaptation of pebble transducers that  handles atoms on input and output. The proof of the lower bound for three atoms will be rather straightforward, because of the strong constraints on how pebble transducers can handle atoms. Later, we will show that lower bounds on pebble transducers with atoms can be automatically lifted to lower bounds without atoms. 

Here is the plan for the rest of this section. In  Section~\ref{sec:atom-lower-bound}, we introduce a variant of pebble transducers that can handle atoms, and we show that for this variant, there are functions of quadratic growth that require any number of pebbles $k$ to be computed. Next, in  Section~\ref{sec:proof-of-deatomization}, we show that the lower bounds with atoms can be lifted to lower bounds without atoms, thus completing the proof of Theorem~\ref{thm:lower-bound}. The lifting result needs to deal with many technicalities, and it is the longest proof in this paper. Nevertheless, we believe that the conceptual essence of the lower bound is captured already in Section~\ref{sec:atom-lower-bound}, which uses the easier setting with atoms. 

\smallskip
\begin{remark}
In this paper, atoms are used to define  computation models for which lower bounds are easier to prove. Another example of this approach can be found in~\cite[Theorem III.1]{bojanczykTuringMachinesAtoms2013}, where it is shown that Turing machines with atoms cannot be determinized (even if one does not care about running time). In the present paper, unlike in~\cite{bojanczykTuringMachinesAtoms2013}, lower bounds  with atoms can be lifted to lower bounds without atoms.
\end{remark}

\input{atom-model}

\input{atom-reduction}
\input{zebra}

%% file: atom-model.tex
\subsection{Pebble transducers with atoms and their lower bounds}
\label{sec:atom-lower-bound}
In this section, we describe an extension of pebble transducers that can  deal with strings that contain atoms. We also prove that for every $k$, there is a function of quadratic growth that needs at least $k$ pebbles to be computed in this model.

\paragraph*{Pebble transducers with atoms.} We begin by describing the model. 
The idea is that the atoms are handled in a very restricted way: the only way to produce an atom in the output is to copy the atom that is underneath the head. This restriction will significantly simplify lower bound proofs.

The model of $k$-pebble transducers is extended to cover atoms in the following way. The input and output alphabets are of the form 
\begin{align*}
\myunderbrace{\Sigma + \atoms}{the input alphabet is \\ 
the disjoint union of \\ some finite set $\Sigma$ \\ and the atoms}
\hspace{2cm}
\myunderbrace{\Gamma + \atoms}{the output alphabet is \\ 
the disjoint union of \\ some finite set $\Gamma$ \\ and the atoms}.
\end{align*}
The letters from the finite alphabets $\Sigma$ and $\Gamma$ will be used to encode formatting symbols, such as separators or brackets.
The transitions are defined by \mso formulas in the same way as without atoms, with the input string viewed as a structure over the vocabulary 
\begin{align*}
\myunderbrace{x \le y}{order on\\ positions}
\qquad 
\myunderbrace{a(x)}{labels for \\ $a \in \Sigma$}.
\end{align*}
In particular, if a position is labeled by an atom, then it satisfies none of the predicates $a(x)$ for $a \in \Sigma$. This means that, unlike for the usual logics for atoms~\cite{segoufin2006automata}, there is  no way of comparing input atoms to each other, in particular, the transducer has no way of checking if two input positions carry the same atom\footnote{The model described here is meant to be a tool in the lower bound proof. It is not meant to be a proposal for polyregular functions on infinite alphabets. Such a proposal would likely involve some mechanism of checking if two input positions carry the same atom.}. To create atoms in the output string, we extend the output mechanism as follows: for each state of the transducer, there is an associated output letter, which is either a  letter from $\Gamma$, or a designated letter called ``atom under the head''. This letter determines the output produced by a configuration with the state, with the designated letter producing the atom under the head.  If the letter under the head is not an atom, or the state has stack height zero and there is no head, then the special letter is replaced by the empty string. 

This completes the definition of pebble transducers with atoms. When we speak of a pebble transducer computing a function that uses atoms in its alphabet, this is the model that we refer to. The rest of Section~\ref{sec:atom-lower-bound} is devoted to lower bounds for this model.

\subsubsection{Atom squaring needs three pebbles}
\label{sec:atomic-map-power}
We begin by explaining how the atom squaring function
\begin{align*}
a_1 \cdots a_n \in \atoms^* \qquad \mapsto  \qquad a_1 a_1 \cdots a_i a_j \cdots a_n a_n
    \end{align*}
can be computed using three pebbles, but not with two. 

Here is a description of the upper bound, i.e.~a three pebble transducer that computes the function. The transducer has six states:
\begin{align*}
\myunderbrace{p_0}{stack\\ height $0$}
\qquad 
\myunderbrace{p_1,q_1}{stack\\ height $1$}
\qquad 
\myunderbrace{p_2,q_2}{stack\\ height $2$}
\qquad 
\myunderbrace{p_3}{stack\\ height $3$}.
\end{align*}
The transducer begins in state $p_0$. Instead of describing the transitions in detail, we show in the following picture a prefix of the accepting run on an input string $123$:
\mypic{33}
The general idea is that the first two pebbles on the stack are used to systematically explore all pairs of input positions in lexicographic order. The purpose of the third pebble is that sometimes we want to output the atom from the first pebble in the stack, but the  model only allows outputting the atom under the head. For this reason, a third pebble needs to be pushed. 

% If the input string is empty, then it ends its computation. Otherwise, it goes to state $p_1$ and pushes the first position onto the pebble stack. When in state $p_1$, the transducer goes to state $p_2$ and pushes the first position onto the pebble stack. When in state $p_2$, the transducer goes to state $p_3$ and pushes a copy of the first pebble onto the top of the stack. Now we can output the atom under the head (this would have corresponded to outputting the atom under pebble one in the configuration with state $p_2$, which is not allowed by the model). After outputting the atom under the head, the transducer pops the head and moves to state $q_2$, in which it outputs the atom under the head. When in state $q_2$, the head is moved to the next position and the state is changed to $p_2$; unless the head was already on the last position, in which case the head is popped and the state is changed to $q_1$. The same is done in state $q_1$: the head is moved to the next position with state $p_1$; unless the head is already at the last position, in which case the computation stops.

We now prove the lower bound.
\begin{lemma}\label{lem:two-pebble-counter-example-atoms}
 The atom square  function
 is not recognized by any pebble transducer that has only two pebbles.
\end{lemma}
\begin{proof} 
    Consider a pebble transducer with two pebbles.
    In a run of this transducer, the number of configurations of stack height one (i.e.~with a pebble stack that has only one pebble) is linear in the input string. By splitting a run along such configurations, we can decompose every run into a linear number of subruns, such that in each  subrun,  pebble one stays fixed and only pebble two can be moved (or is not present). To complete the proof of the lemma, we will show that each subrun can produce an output of at most constant size, and therefore the entire output of the pebble transducer can be at most linear, and thus shorter than the output of atom squaring.

    Consider then a subrun  where pebble one is fixed, and the second pebble is moving. In this run, the head can visit each position at most once per state, and therefore each atom can be repeated in the output at most a constant number of times, because an atom is output only when it is under the head. (Here, we assume that all atoms in the input string are distinct.) If the input to atom squaring has $n$ letters, then in the output string the first coordinate is changed at most once every $n$ positions, and therefore the output size for a run where pebble one is fixed cannot exceed a fixed constant. 
 \end{proof}

%% file: zebra.tex
\subsubsection{Alternating squaring}
\label{sec:zebra}
In Section~\ref{sec:atomic-map-power}, we presented a quadratic function with atoms that needs three pebbles to be computed. In this section, we strengthen the lower bound to an arbitrary number of pebbles, as stated in the following lemma.
\begin{lemma}\label{lem:zebra-atoms}
    For every $k \in \set{1,2,\ldots}$ there is a  function of quadratic growth (with atoms) that can be computed by a pebble transducer that uses  $2k+1$ pebbles, but not by one that uses $2k$ pebbles.
\end{lemma}

In the proof of the lemma, it will be more convenient to think of the inputs and outputs as being trees of bounded height; these trees can then be represented as strings to get a string-to-string function as required by the lemma. 

The lemma will be witnessed by transductions that are based on a tree operation, called \emph{alternating product}.  For two trees $s$ and $t$, their 
 {alternating product} is defined as follows by induction on the height of $t$. When the height of $t$ is nonzero, then the alternating product is the tree whose root label is the pair (root label of $t$, root label of $s$), and where the child subtrees are all trees that are obtained by taking the alternating product of $s$ with some child subtree of $t$ (listed in the same order as the children of $t$). When the height of $t$ is zero, i.e.~$t$ is just one node, then the root of the alternating product is defined in the same way, and there are no other nodes. Define the \emph{alternating square} of a tree to be the alternating product of the tree with itself. Here is a picture of a tree and its alternating square.
\mypic{30}
The alternating square operation doubles the height of the input tree.  We will only apply this operation to trees which are balanced, i.e.~all root-to-leaf paths have the same length. In this case, the leaves of the output tree  are exactly the pairs of leaves of the input tree.

We will prove the lemma by using the following function: the input is a balanced tree of height $k$ with nodes labeled by atoms, and the output is its alternating square. To view this function as a string-to-string operation, we encode trees as strings, as in the following example based on the picture above:
\begin{eqnarray*}
% &\myunderbrace{1 \langle 2 \langle 34 \rangle 5 \langle 678 \rangle \rangle
% }{string representation\\ of the input tree}
% \\ 
&\myunderbrace{11 \langle  12 \langle 22 \langle 23 \langle 33\ 34 \rangle 24 \langle 43\ 44 \rangle  \rangle  \cdots 58 \langle 86\ 87\ 88 \rangle  \rangle \rangle }{string representation \\ of the output tree}.
\end{eqnarray*}
If the input string is not well formatted, i.e.~it does not represent a balanced tree of height $k$ labeled by atoms, the output of the transducer is the empty string.  It is not hard to see that in the case of $k=1$, we essentially encounter the atomic squaring function, which needed $2k+1 = 3 $ pebbles.

We have already proved that alternating squaring has quadratic growth. To complete the proof of the lemma, we will show that if the inputs are trees of height $k$, then  the function  can be computed using $2k+1$ pebbles, but it not using $2k$ pebbles. The upper bound of $2k+1$ pebbles is straightforward.  The run of the transducer corresponds to a program with $2k$ nested loops as explained below  (the lines in the code coloured \red{red} and \cyan{blue} to underline the alternating character of the loops):

\begin{center}
    \begin{tabular}{l}
        \red{$x_0$ := root}\\
        \cyan{$y_0$ := root}\\
        % \red{\hspace{0.1cm} output $\treemarker 0$}\\
        \red{\hspace{0.1cm} for $x_1$ in children of $x_0$}\\
        % \cyan{\hspace{0.2cm} output $\treemarker 1$}\\
        \cyan{\hspace{0.2cm}  for $y_1$ in children of $y_0$}\\
        % \red{\hspace{0.3cm} output $\treemarker 2$}\\
        \red{\hspace{0.3cm} for $x_2$ in children of $x_1$}\\
        % \cyan{\hspace{0.4cm} output $\treemarker 3$}\\
        \cyan{\hspace{0.4cm}  for $y_2$ in children of $y_1$}\\
        \hspace{0.5cm} \dots\\
        % \red{\hspace{0.6cm} output $\treemarker {2k-2}$}\\
        \red{\hspace{0.6cm} for $x_k$ in children of $x_{k-1}$}\\
        % \cyan{\hspace{0.7cm} output $\treemarker {2k-1}$}\\
        \cyan{\hspace{0.7cm}  for $y_k$ in children of $y_{k-1}$}\\
    % \red{\hspace{0.8cm} output label of $x_k$}\\
    % \cyan{\hspace{0.8cm} output label of $y_k$}
    \end{tabular}
\end{center}

Although the program has $2k$ nested loops, a closer inspection reveals that it requires $2k+1$ pebbles to be implemented. The reason is the same as for atomic squaring: in the innermost loop, the program needs to output the pair of labels of the two positions $\red{x_k}$ and $\cyan{y_k}$. The head of the pebble transducer is at the second position  $\cyan{y_k}$, and therefore, in order to output the first position $\red{x_k}$ we need to push another pebble, due to the output mechanism of our transducer model.

The rest of this section is devoted to showing the lower bound, i.e.~that $2k$  pebbles are not enough to compute the alternating square for input trees of height $k$. In the proof, we will show that certain runs can only touch small parts of the output tree, in the following sense.  Recall that all runs considered are parts of an accepting run, and therefore each output symbol produced by a configuration can be attributed to a unique node in the output tree. When we say that a run \emph{touches} a subtree of the output tree, we mean that at least one configuration in the run produces an output that is attributed to this subtree. 

We will show that the output of a run is bounded by a parameter that is related to the tree structure of configurations, as explained below. 
For a configuration $c$, its \emph{descendants} are defined to be all configurations that appear strictly between $c$ and the next configuration that has the same or lower stack height than $c$. If the input is fixed,  the descendant relation imposes a tree structure on the configurations; we will use the name \emph{tree of configurations} for this tree. Here is a picture of the tree of configurations for thee transducer from Section~\ref{sec:atomic-map-power}:
\mypic{39}   Define the \emph{height} of a run to be the height of the smallest subtree in the configuration tree that contains this run.

% For each run, its output covers part of the output tree. The following lemma will prove that a $2k$ pebble transducer cannot compute the $k$-zerbra function, by showing that balanced runs of height $\ell$ can cover a part of the output tree that is small in the following sense: it  touches at most a bounded number of subtrees of the output tree that have height $\ell-1$. In  particular, since the entire accepting run consists of a bounded number of balanced runs of height at most $2k$, it follows that only a bounded number of subtrees of the output that have  height $2k-1$ can be produced, which means that the entire output tree  of height $2k$ cannot be produced; thus proving that a $2k$ pebble transducer cannot compute the $k$-zebra function, as required by Lemma~\ref{lem:zebra-atoms}. 

\begin{lemma}\label{lem:zebra-proof}
    For every $\ell \in \set{1,\ldots,2k}$ there are constants $c(\ell), d(\ell) \in \set{0,1,\ldots}$ with the following property. Consider an input  tree to the $k$-alternating square function, where all atoms are pairwise different, and which has degree at least $d(\ell)$, which means that all non-leaf nodes have at least  $d(\ell)$ children. If  a run  over this input tree has height $\ell$, then it touches at most  $c(\ell)$ subtrees of the output tree that have height  $\ell-1$.
\end{lemma}

A corollary of the lemma is  that the entire accepting run, which is a run of height $2k$, can touch only a constant number of subtrees of height $2k-1$, and thus it cannot produce the  entire output for the  $k$-alternating square function, thus completing the proof of Lemma~\ref{lem:zebra-atoms}.  It remains the prove the lemma.

\begin{proof}[Proof of Lemma~\ref{lem:zebra-proof}]
    Induction on $\ell$. The induction base of  $\ell=1$ is proved in the same way as in Lemma~\ref{lem:two-pebble-counter-example-atoms}. In the output tree, every subtree of height $1$ will have the same atom repeated in all of its leaves. For a run of height $\ell=1$, the head can be in each position at most once per state, and therefore if the degree of the input tree exceeds the number of states, the run can touch only a constant number of leaves in each subtree of the output tree.

    We now present the induction step, where we prove the lemma for $\ell +1$, assuming that it is true for $\ell$. In the squaring function, there is an injective correspondence, which maps each subtree of the output tree to a pair of subtrees in the input tree, this correspondence is illustrated in the following picture:
    \mypic{31}
    For a subtree of the output tree, the corresponding pair of subtrees in the input tree is called its \emph{origin pair}. The origin pairs are exactly those pairs of subtrees in the input tree where the first coordinate has height that is equal to, or bigger by one than, the height of the second coordinate.  Since an origin pair represents exactly one subtree of the output tree, notions about subtrees of the output can also be applied to the corresponding origin pairs. For example, we say that a run touches an origin pair if it touches the corresponding  subtree in the output.  Likewise we define the height of an origin pair to be the height of the corresponding subtree in the output tree; this is the same as the sum of the heights of the subtrees of the input tree that appear in the origin pair.

    We prove the induction step by contradiction: we will show that if the lemma fails for $\ell+1$, then it fails for $\ell$. In the following claim, we show that a failure for $\ell+1$ implies that runs contain certain large patterns. The patterns are called \emph{$n\times n$ squares}; these are sets of pairs of trees of the form $X \times Y$ where both $X$ and $Y$ have  size $n$.

\begin{claim}\label{claim:large-squares} If the lemma fails for $\ell+1$, then for every $n$ there is a run of height $\ell+1$ such that the input tree has degree at least $n$, and the set of origin pairs of height $\ell-1$ touched by this run contains an $n \times n$ square.
\end{claim}
\begin{proof}
    If the lemma fails for $\ell+1$, for every $n$ we can find a run of height $\ell+1$ such that the input tree has degree at least $n$, and the run touches at least $n$ subtrees of the output with height $\ell$. Apply this observation to $3n$, yielding a run of height $\ell+1$ where the input tree has degree at least $3n$ and the run touches at least $3n$ subtrees of the output tree that have height $\ell$.  Consider the list of these at least $3n$ subtrees,  listed in the order that they are touched. Partition this list into intervals, in which subtrees from the same interval are consecutive, i.e.~their roots are siblings. Since the input tree has degree at least $n$, the output tree also has  degree at least $n$, and therefore each interval from the partition, with the possible exception of the first and last intervals, has length least $n$. This means that if the list has length at least $3n$, then some interval has length at least $n$. Summing up, we know that the run must touch at least $n$ consecutive subtrees of the output that have height $\ell$. Let the origin pairs of these consecutive subtrees be 
    \begin{align*}
    (s,t_1),\ldots,(s,t_n).
    \end{align*}
    These pairs share the same first coordinate, because siblings in the output tree have origin pairs that share the first coordinate. The origin pairs touched  by the run will therefore contain the following set
\begin{align*}
\text{children of $s$} \times \set{t_1,\ldots,t_n},
\end{align*}
which consists of height $\ell-1$ origin pairs, and  contains an $n \times n$ square by the assumption on the degree of the input tree being at least $n$.
\end{proof}

In the conclusion of the claim above, we have an $n \times n$ square of origin pairs of height $\ell-1$ inside a run of height $\ell+1$. Inside that run we will find a run of smaller height $\ell$ which uses a number of these pairs that is linear in $n$ and therefore arbitrarily large; thus proving that the lemma fails for $\ell$ and completing the induction step. To prove this, we use the following observation  about squares definable in \mso. 

\begin{claim}\label{claim:squares} Let $\varphi(x,y,z)$ be an \mso formula which selects triples of positions in strings. There is some $\lambda > 0$ with the following property. For every input string, if there  is an $n \times n$ square contained in the set of pairs $(x,y)$ which satisfy
    \begin{align*}
    \exists z \ \varphi(x,y,z),
    \end{align*}
then  for some position $z$,  there are at least $\lambda n$ pairs $(x,y)$ satisfy $\varphi(x,y,z)$. 
\end{claim}
\begin{proof} 
    Consider an input string  in which there is an $n \times n$ square of the form $X \times Y$ as in the assumption of the claim. For each pair there is some witness $z$.  Define the \emph{type} of a witness $z$ to be the \mso theory of this witness with respect to the distinguished positions $X \cup Y$. This type is uniquely determined by  the input string, the order relationship of $z$ with the distinguished  positions, and some fixed regular information about the parts of the string between $z$ and the nearest distinguished positions on the left and right. In particular, once the input string is fixed, the possible number of types is at most $cn$ for some constant $c$ that depends only on the formula. It follows that for at least $n^2/cn = n/c$ 
    pairs in the square, the corresponding witnesses have the same type. Witnesses with the same type  can be swapped, thus proving the claim. 
\end{proof}

We now use Claims~\ref{claim:large-squares} and~\ref{claim:squares} to complete the proof of the induction step. In the proof, it will be more convenient to discuss special runs, called balanced runs. A run is called \emph{balanced} if it arises by taking some configuration and all of its descendants in the tree of configurations. In other words, we take a configuration and continue the run until, but not including, the nearest configuration with  the same or smaller number of pebbles. By definition, balanced runs are in one-to-one correspondence with configurations; therefore we can apply to balanced runs notions that are defined for configurations, such as the child relation from the tree of configurations, or the position of the head.  
Consider  a balanced run $\rho$ of height $\ell+1$. We represent this run as a string over a finite alphabet in the following way:
\mypic{29}
For a balanced run, consider the following property 
\begin{align*}
\varphi(x,y,z)
\end{align*}
of nodes in the input tree:  the pair of subtrees with roots in $x$ and $y$ is an origin pair of height $\ell-1$ that is touched by some child of the run that has head  position $z$. 
Using the above string representation,  this relation  on input positions can be  formalized in \mso. By definition, 
\begin{align}\label{eq:exists-z}
    \exists z \ \varphi(x,y,z)
    \end{align}
describes exactly the set of origin pairs that have height $\ell-1$ and are touched by  the run with its first configuration removed. This set  is the same as  the set of pairs in the conclusion of Claim~\ref{claim:large-squares} with one pair removed, and therefore we can apply that  claim to conclude that   if the lemma would fail for $\ell+1$, then for every $n$ one could find a run of height $\ell+1$ such that  the set in~\eqref{eq:exists-z} contains an $n \times n$ square. By Claim~\ref{claim:squares}, there would be some position $z$ in the input tree  that admits linear in $n$ number of pairs $(x,y)$ which satisfy $\varphi(x,y,z)$. In other words, there is some position $z$ such that there is a linear in $n$  number of  origin pairs of height $\ell-1$ that are touched by children  with their head in position $z$.  Children of the run have height $\ell$, and since a position $z$ can be used as the head  for at most one child  per state, this would mean that run of height $\ell$ touches a linear in $n$ number of 
subtrees of the output that have height $\ell-1$. This means that the lemma fails for $\ell$,  thus completing the proof of the induction step.
\end{proof}

%% file: deatomization-proof.tex
\subsection{Deatomization}
\label{sec:proof-of-deatomization}
 In this section, we show that the lower bounds with atoms, such as the lower bound proof in Lemma~\ref{lem:two-pebble-counter-example-atoms} or~\ref{lem:zebra-atoms}, can be lifted to lower bounds without atoms, thus completing the proof of Theorem~\ref{thm:lower-bound}. This lifting result is the most technical part of the proof of Theorem~\ref{thm:lower-bound}, and it shows that for each pebble transducer with atoms there is a corresponding pebble transducer without atoms which requires the same number of pebbles as the original one. 

In the proof, we use two important properties of a function
that is computed by a pebble transducer with atoms. Intuitively speaking, these are: (a) the function can only move around or duplicate atoms from the input string, but it cannot compare them to each other; and (b) if atoms are represented by strings over a finite alphabet, then the function can be implemented by a pebble transducer without atoms, using the same number of pebbles. The main result of this section will be that these properties are not only necessary, but they are also sufficient. 

We begin by describing the two properties in more detail.

\paragraph*{Atom-oblivious functions.}
The first condition, about not comparing atoms to each other, will be abstracted by saying that the function commutes with all functions from atoms to atoms. Consider a function 
\begin{align*}
 f : (\Sigma + \atoms)^* \to 
 (\Gamma + \atoms)^*,
\end{align*}
i.e.~a function whose inputs and outputs use atoms and letters from a finite alphabet (as is the case for functions computed by pebble transducers with atoms).
We say that the function  is \emph{atom oblivious} if the  diagram 
\[
\begin{tikzcd}
(\Sigma + \atoms)^* 
\ar[r,"f"]
\ar[d,"\pi"']
&
(\Gamma + \atoms)^* 
\ar[d,"\pi"] \\
(\Sigma + \atoms)^* 
\ar[r,"f"']
&
(\Gamma + \atoms)^* 
\end{tikzcd}
\]
commutes for every input string $w$ and every function $\pi : \atoms \to \atoms$, not necessarily bijective\footnote{Here we consider functions that are not necessarily bijections. If we only require commuting with bijective $\pi$, then the resulting property is called \emph{equivariance} and it is the central property in sets with atoms.}. In the diagram above, the vertical arrows use the natural extension of $\pi$ from atoms to strings that use atoms. The general idea behind atom-oblivious functions is that they are allowed to move around or copy atoms from the input string, but they are not allowed to read them or compare them in any way. By design, any function computed by a pebble transducer with atoms will be atom-oblivious.

\paragraph*{Deatomization.} We now turn to the second property, which is that pebble transducers with atoms can be simulated by pebble transducers over finite alphabets, assuming a representation of atoms by strings over a finite alphabet. We use the representation explained in the following picture:
\mypic{22}
The brackets $\langle$ and $\rangle$ in the above representation, as well as the letter $a$ (which will be called the unit letter) are fresh, and should not be confused with any other symbols that might appear in the alphabets $\Sigma$ and $\Gamma$, e.g.~the brackets used to represent the tree structure in the proof of Lemma~\ref{lem:zebra-atoms}. 
The representation is parameterized by some injective function that maps atoms to atom blocks, i.e.~strings in $\langle a^* \rangle$. Such a function will be called an \emph{atom representation}. 
Throughout this section, we use a colour convention where  red is used for strings to which an atom representation has been applied, i.e.~red variables denote words in which atoms are represented using atom blocks.
 Define a \emph{deatomization} of the function $f$ to be any function $\red f$ (here we use the colour convention) which makes the following diagram commute for every atom representation $\alpha$:
 \[
 \begin{tikzcd}[column sep=3cm]
 (\Sigma + \atoms)^*
 \ar[r,"f"]
 \ar[d," \alpha"'] &
 (\Gamma + \atoms)^*
 \ar[d," \alpha"]\\
 (\Sigma + \langle a^* \rangle)^* 
 \ar[r,red,"\red f"']&
 (\Gamma + \langle a^* \rangle)^*
 \end{tikzcd}
 \]

 \begin{fact}
    If $f$ atom-oblivious, then it has a unique deatomization.
 \end{fact} 
 \begin{proof}
    The unique deatomization works as follows. Given an input string for the deatomization, replace every atom block with a distinct atom (if the same atom block has several occurrences in the input string, a different atom is used for each occurrence), then apply $f$, and finally replace each atom from the output string with the corresponding atom block.  By the assumption on atom-obliviousness, this is the only way that the de-atomization can work. This is because an atom-oblivious function is uniquely defined by the outputs that it produces on inputs in which no atom is used twice.
 \end{proof}

Thanks to the above fact, for atom-oblivious functions we can speak of \emph{the} deatomization.
 Again, it is easy to see that for every pebble transducer with atoms, its deatomization is computed by a pebble transducer without atoms that has the same number of pebbles. The transducer without atoms simply copies the atom block next to the head whenever the transducer with atoms wishes to output that atom.

\paragraph*{The theorem.} As we have remarked above, if a function is computed by a $k$-pebble transducer with atoms, then it is atom oblivious and its deatomization is computed by a $k$-pebble transducer without atoms. The main result of this section is that the implication is in fact an equivalence. 
\begin{theorem}[Deatomization]\label{thm:deatomization}
 A function 
 \begin{align*}
 f : (\Sigma + \atoms)^* \to 
 (\Gamma + \atoms)^*
 \end{align*}
 is computed by a $k$-pebble transducer with atoms if and only if it is atom-oblivious, and its deatomization is computed by a $k$-pebble transducer without atoms.
\end{theorem}

The above theorem, which is proved in the appendix, completes the proof of Theorem~\ref{thm:lower-bound} about quadratic polyregular functions needing arbitrarily large pebble stacks.  

\begin{proof}[Proof of Theorem~\ref{thm:lower-bound}, assuming the Deatomization Theorem]
    Consider the function from Lemma~\ref{lem:zebra-atoms}. As we have shown, this function has quadratic growth, but it requires at least $2k+1$ pebbles with atoms. Therefore, thanks to the  Deatomization Theorem, its deatomization  also requires at least $2k+1$ pebbles. It is also easy to see that this deatomization has quadratic growth.     
\end{proof}

%% file: ap-quant-elim.tex
\section{Quantifier elimination}

In this part of the appendix, we prove Theorem~\ref{thm:quantifier-elimination}, about quantifier elimination for \mso interpretations.

    This theorem is a straightforward corollary of the Factorization Forest Theorem and compositionality of \mso. 
    
    Consider an \mso interpretation that defines the function $f$. Let $\Phi$ be the set of \mso formulas that appear in this interpretation, either as universe formulas or as formulas defining relations of the output structure. We use the following standard result about \mso on strings, which we refer to as \emph{compositionality}.
    
    \begin{lemma}\label{lem:compositionality}
      Let $\Phi$ be a set of \mso formulas, which may have free first-order variables, over the vocabulary of strings over some input alphabet $\Sigma$.
       There is a monoid homomorphism 
       \begin{align*}
       h : \Sigma^* \to M
       \end{align*}
       into a finite monoid, such that for every \mso formula 
       \begin{align*}
       \varphi(x_1,\ldots,x_\ell) \in \Phi,
       \end{align*}
       whether or not a string in $\Sigma^*$ with $\ell$ distinguished positions satisfies the formula depends only on the following information: (a) the order of the distinguished positions and their labels; (b) the values of the homomorphism on the intervals in the input string that separate distinguished positions, as explained in the following picture:
          \mypic{32}
    \end{lemma}
    
    We use factorization trees for the homomorphism from the above lemma, defined as follows.
    Recall that an \emph{idempotent} is a monoid element $e \in M$ such that $ee=e$. Define a \emph{factorization tree} to be a tree where: 
    \begin{itemize}
    \item every leaf is labeled by a letter from $\Sigma$;
    \item every node that is not a leaf is labeled by the value of the homomorphism on the yield of the subtree of the node;
    \item for every node that has at least three children, there is some idempotent $e$  such that the node and all of its children have label $e$.
    \end{itemize}
    By the Factorization Forest Theorem, there is some $k$ such that every string is the yield of some factorization tree of height at most $k$. Let $\Tt$ be the factorization trees of height at most $k$; this is easily seen to be a tree grammar. Furthermore, a factorization tree can be computed by a linear interpretation~\cite[Section 4]{bojanczykFactorizationForests2009}, which gives us the left part of the diagram in the theorem: 
    \[
    \begin{tikzcd}
    [column sep=1.5cm]
    \Sigma^* 
    \ar[r,"\text{linear }h"]
    \ar[dr,equals, bend right=30]
    & \Tt 
    \ar[d, "\text{yield}"']
    \\
    & 
    \Sigma^*.
    \end{tikzcd}
    \] 
   
   The homomorphism $h$ was chosen so that for every formula defining the interpretation $f$, whether or not this formula selects a tuple of distinguished positions can be determined by the order of the variables, their labels, and the values of $h$ on the infixes separating the distinguished positions. All of this information can be recovered by using quantifier-free formulas in the Simon structure of the factorization tree, see~\cite[Proof of Theorem 2]{bojanczykFactorizationForests2009}. Therefore we get the remaining part of the theorem, namely 
    \[
    \begin{tikzcd}
    [column sep=1.5cm]
    \Tt 
    \ar[dr,"\text{quantifier-free }g"]
    \ar[d, "\text{yield}"']
    \\ 
    \Sigma^*
    \ar[r,"f"']
    &
    \Gamma^*.
    \end{tikzcd}
    \] 

%% file: ap-proof-of-basis.tex
\section{Proof of the Basis Lemma}
\label{sec:proof-the-basis-lemma}
In this part of the appendix, we prove the Basis Lemma. We do this using a syntactic analysis of a tree that corresponds to each quantifier-free type. Consider a tree $t \in \Tt$ and a tuple of distinguished nodes in this tree. Define the \emph{skeleton} of this tuple of nodes to be the structure that arises by restricting the original tree to the distinguished nodes and their ancestors.  The skeleton inherits the distinguished nodes from the original tree, and it inherits the relations from the original tree. The isomorphism type of the skeleton is the same thing as the quantifier-free theory of the distinguished nodes.  It is important that in the skeleton, the relations for successor sibling, leftmost sibling and rightmost sibling are inherited from the original tree. For example, if a node $x$ is in the skeleton, but all of its siblings to the left in the original tree are not in the skeleton, then $x$ will  not be selected by the unary relation ``leftmost sibling'' in the skeleton, despite not having left siblings in the skeleton. Similarly, there might be two nodes that are successor siblings in the skeleton, but which are not connected by the ``successor sibling'' relation, because the separating nodes were deleted when going to the skeleton.  

\begin{myexample}\label{ex:skeleton-continued}
    Here is a picture of a skeleton
    \mypic{35}
    In the picture above, the ellipses $\cdots$ represent deleted nodes, which describes to the relations ``successor sibling'', ``leftmost sibling'' and ``rightmost sibling'' in the skeleton.  For example the nodes $x_5$ and $x_6$  in the skeleton are not selected by the ``successor sibling'' relation, even though they are not separated in the skeleton by any other node in the sibling order. 
\end{myexample}

Let $\phi(x_1,\ldots,x_k)$ be some quantifier-free type, as in the assumption of the Basis Lemma. This quantifier-free type is the same thing as a skeleton (modulo isomorphism of skeletons). We will prove  the Basis Lemma by a syntactic analysis of the skeleton, similar to the analysis in Example~\ref{ex:dependencies-in-skeleton}.

\begin{definition}[Dependency graph]\label{def:dependency-graph}
    For a skeleton, its \emph{dependency graph} is the directed graph where the vertices are  nodes of the skeleton, and there is a directed edge $x \to y$ if any of the following conditions hold:
    \begin{enumerate}
        \item $y$ is the parent of $x$; or
        \item  $y$ is a child of $x$  selected by ``leftmost sibling''; or
        \item $y$ is a child of $x$  selected by ``rightmost sibling''; or 
        \item $x$ and $y$ are selected by  ``successor sibling''.
    \end{enumerate}
\end{definition}

Note that the vertices are all nodes of the skeleton, which includes the distinguished nodes (corresponding to the free variables in a quantifier-free type), and their ancestors. Also, the relations ``leftmost sibling'', ``rightmost sibling'' and ``consecutive sibling'' in the above definition are  inherited from the original  tree, and need not describe the relationship between nodes that are in the skeleton, as discussed in Example~\ref{ex:skeleton-continued}. 

\begin{myexample}
    Here is the dependency graph for the skeleton from Example~\ref{ex:skeleton-continued}.
    \mypic{36}
    In the dependency graph, we will be mainly interested in the \emph{minimal \scc's}, which are strongly connected components that cannot be reached from any other strongly connected components. 
    Note that every minimal \scc contains at least one variable (i.e.~at least one distinguished node). This is because every node in the skeleton is either a distinguished node, or an ancestor of some distinguished node. 
\end{myexample}

As remarked in the above example, every minimal \scc in the dependency graph contains at least one variable. For each minimal \scc choose exactly one variable, yielding a subset of the variables  
\begin{align*}
    X \subseteq \set{x_1,\ldots,x_k}.
    \end{align*}
We will prove that this subset satisfies the two conditions in the Basis Lemma.

Consider first Condition~\ref{basis:spans}, which says that the variables from $X$ span all the other variables. 
Each edge in the dependency graph describes a functional dependency. Therefore, we can see that if there is a path in the deependency graph from some variable $x_i$ to some variable $x_j$, then for every every tree $t \in \Tt$, if two $k$-tuples selected by $\varphi$ agree on variable $x_i$, then these tuples must also agree on $x_j$. Since every variable admits a path from some variable in $X$, because $X$ represents all minimal \scc's, it follows that the variables from $X$ determine the other variables, are required by Condition~\ref{basis:spans}.

It remains to prove Condition~\ref{basis:spans}. Let $k$ be the size of the basis $X$. We need to show that there is a sequence of trees
\begin{align*}
t_1,t_2,\ldots \in \Tt
\end{align*}
such that the  tree $t_n$  has size $\Oo(n)$, while the number of tuples selected by the quantifier-free type $\varphi(x_1,\ldots,x_k)$ is at least $n^k$. 
The tree $t_n$  is constructed as follows. For every minimal \scc in the dependency graph, create $n$ copies which are attached to the same parent, as explained in Figure~\ref{fig:copy-skeleton}. Next, for every ellipsis in the resulting skeleton, insert some node subject to the constraints on labels in the tree grammar. The resulting tree is $t_n$. It is easy to see that its size is linear in $n$, since we copy $n$ times a constant number of patterns of constant size. By construction, for each of the $k$ minimal \scc's in the dependency graph, we can independently assign the corresponding variables to at least $n$ possible parts in $t_n$, which gives growth that is at least $n^k$. This completes the proof of Condition~\ref{basis:spans}, and therefore also of the Basis Lemma. 

\begin{figure}
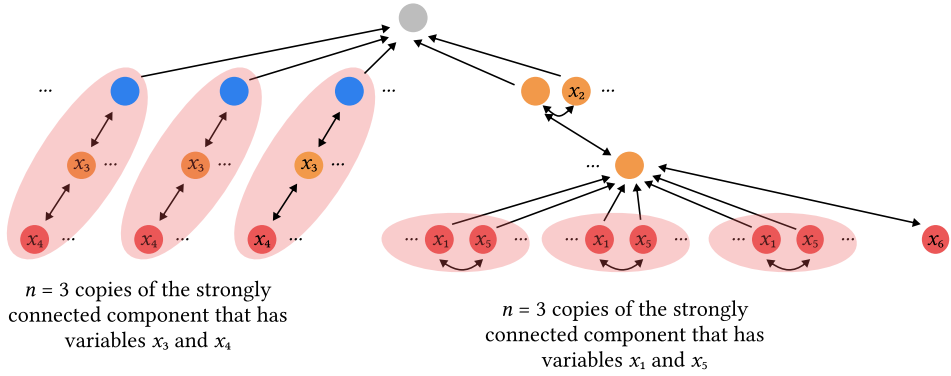

    \mypic{38}
    \caption{\label{fig:copy-skeleton} Taking $n$ copies of every minimal \scc in a skeleton}
\end{figure}

%% file: equivalent-pebble.tex
\section{Equivalent models of pebble transducer}
\label{sec:equivalent-pebble-transducers}

In this part of the appendix, we prove Lemma~\ref{lem:equivalent-pebbles}, which says that for every number of pebbles $k \in \set{1,2,\ldots}$, the model from Definition~\ref{def:pebble-transducer} computes the same string-to-string functions as the model  defined in~\cite[Section 1]{engelfriet2002two}.  
For the purpose of this section, we use the name \emph{\mso pebble transducer} for the model defined in Definition~\ref{def:pebble-transducer}, and the name \emph{classical pebble transducer} for the model from~\cite{engelfriet2002two}.  The former model is the one that is used in this paper, in particular the lower bounds are proved for it.

\paragraph*{Definition of the classical model.}
We begin by defining the classical model. The following definition is easily seen to be equivalent to the one from~\cite[Section 1]{engelfriet2002two}, with the main difference being our way of counting pebbles: since we count the head as a pebble and~\cite{engelfriet2002two} does not, see Footnote~\ref{footnote:headcount}, what  we call a $k$ pebble transducer here is  called $k-1$ pebble transducer in~\cite{engelfriet2002two}. The other difference is that~\cite{engelfriet2002two} uses endmarkers $\triangleright$ and $\triangleleft$ to delimit the input string, while the definition below uses tests that tell us if the head is on an extremal position. 

\begin{definition}[Classical pebble transducer]\label{def:classical-model}
    The syntax of a \emph{classical pebble transducer} is given by
\begin{itemize}
    \item a number of pebbles $k \in \set{1,2,\ldots}$;
    \item a finite set $Q$ of states with a designated initial state;
    \item finite input and output alphabets  $\Sigma$ and $\Gamma$;
    \item a designated output string in $\Gamma^*$ for the empty input;
    \item a transition function 
    \begin{align*}
    \delta : Q \times \myunderbrace{\powerset(\text{tests})}{which tests are satisfied  by \\ the  present pebble stack}  \to Q \times \Gamma^* \times \myunderbrace{\text{actions}}{actions that modify \\ the  pebble stack},
    \end{align*}
    where the  tests and actions are defined as follows:
    \begin{itemize}
        \item Tests. In the following tests, the numbers $i,j$ refer to pebble names in $\set{1,\ldots,k}$:
        \begin{itemize}
            \item is pebble $i$ defined, i.e.~present in the stack?
            \item do pebbles  $i,j$ point to the same input position?
            \item does pebble $i$ point to the leftmost input position?
            \item does pebble $i$ point to the rightmost input position?
            \item does pebble $i$ point to a position with label $a \in \Sigma$?
        \end{itemize}
        \item Actions. 
        \begin{itemize}
            \item stop;
            \item move the head one step to the left;
            \item move the head one step to the right;
            \item pop the topmost pebble from the stack;
            \item push a new pebble to the stack, pointing to the leftmost position. 
        \end{itemize}
    \end{itemize}
\end{itemize}
\end{definition}

The semantics of the transducer is a partial function of type $\Sigma^* \to \Gamma^*$ that is defined as follows. If the input string is empty, then the output is the designated output string given in the syntax. Otherwise, the transducer begins in a configuration that consists of the initial state and a pebble stack that has height one, with the unique pebble (the head) pointing to the lefmost input position. Next, it starts to update the configuration, according to the transition function, with each transition appending some string to the output. The actions in the transition function might fail: the head might be moved outside the input string, the transducer might try to push a pebble when the stack has maximal height $k$, or it might try to pop a pebble when the stack has minimal height $1$. If an action fails, then the ouptut string is undefined. The run might  enter an infinite loop, in which case the output string is also undefined.  This completes the semantics of the classical model. 

We now show that the classical model described above defines the same total functions as the \mso model from Definition~\ref{def:pebble-transducer}. The inclusion 
\begin{align*}
\text{classical model} \subseteq \text{\mso model}
\end{align*}
is standard, and proved as in~\cite[Lemma 2.3]{bojanczykPolyregularFunctions2018}. We concentrate on the opposite inclusion 
\begin{align}\label{eq:opposite-inclusion}
    \text{classical model} \supseteq \text{\mso model}.
    \end{align}
This inclusion was proved in the case of $k=1$ in \cite[Lemma 6]{engelfrietMSODefinableString2001}, and we explain below how the case of $k>1$ reduces to the case of $k=1$. Before presenting the reduction, we remark that it is not really important in the scope of this paper: our main contribution is lower bounds which work for the \mso model, and therefore the same lower bounds will clearly work for the classical model.

\begin{proof}[Proof sketch for~\eqref{eq:opposite-inclusion}]
   We pass through an intermediate model, in which \mso transitions are only allowed for configurations of maximal height $k$. Define the \emph{intermediate model} to be the model where for configurations of maximal height $k$, the next configuration is determined using \mso as in Definition~\ref{def:pebble-transducer}, and for configurations of non-maximal height $<k$, the  next configuration is determined as in the classical model, i.e.~based on the tests given in Definition~\ref{def:classical-model}. We will prove two inclusions:
   \begin{align*}
    \text{classical model} \supseteq  \text{intermediate model} \supseteq  \text{\mso model}.
    \end{align*}
    
    The second inclusion is proved using compositionality of \mso in a standard way. The idea is that if a configuration has non-maximal height, then the extra pebble can be used to compute appropriate \mso theories, and thus compute the next configuration. 
    
    We now consider the first  inclusion, i.e.~the intermediate model is contained in the classical model. Consider a pebble transducer as in the intermediate model. For a configuration  of almost maximal height $k-1$, consider the subcomputation that is strictly between this configuration and the nearest configuration of height $<k$. In this subcomputation,  which may be empty, the first $k-1$ pebbles are fixed, and the only pebble that is moved is the maximal pebble $k$. Therefore, this subcomputation can be seen as a computation of a one pebble transducer, in which the input string is additionally marked by the fixed positions of the first $k-1$ pebbles.  Using  
    the result from~\cite[Lemma 6]{engelfrietMSODefinableString2001}, this subcomputation can be simulated by a pebble transducer in the  classical model, without \mso transitions. Substituting this transducer for the subcomputation, we get the desired pebble transducer that does not use \mso transitions at all.
\end{proof}

%% file: ap-deatomization.tex
\section{Proof of the Deatomization Theorem}
This section is devoted to proving the hard implication of the Deatomization Theorem.  This implication says that if a function 
\begin{align*}
f : (\Sigma + \atoms)^* \to 
(\Gamma + \atoms)^*
\end{align*}
 is atom-oblivious, and its deatomization is computed by a $k$-pebble transducer without atoms, then the function is computed by a $k$-pebble transducer with atoms.

The proof uses a detailed analysis of how a pebble transducer can output an atom block $\langle a^n \rangle$. In the proof, we use a slightly stronger model of pebble transducer without atoms, in which each configuration is associated to a possibly empty string over the finite output alphabet. Since the model is stronger, the result is stronger: we show that even the stronger model can be de-atomized. The stronger model will be convenient in the proof below, where we gradually improve a transducer so that it satisfies more and more properties.

In this proof, we define a \emph{run} of a pebble transducer to be a sequence of configurations over the same input string, which form an interval in the order of configurations, i.e.~these are all configurations between the first and last one in the sequence. 
 Define an \emph{atom run} to be a run which produces such an atom block, i.e.~an atom run is one which outputs an atom block, with the first configuration producing the opening bracket and the last configuration producing the matching closing bracket. We  prove the Deatomization Theorem in three steps. In Section~\ref{sec:to-honest-first}, we show that a pebble transducer can be improved so that in every atom run, only the head (and not any other pebble below the head) is moved, and furthermore, the head visits only a constant number of atom blocks in the input string. This will be proved using the quantifier elimination techniques from Section~\ref{sec:quantifier-elimination}.  Next, in Section~\ref{sec:to-honest-second}, we further improve the transducer so that the head visits only one atom block in the input. This will be proved using an analysis of certain affine functions that appear implicitly in a pebble transducer. Finally,  in Section~\ref{sec:final-proof-of-deatomization}, we use the improved transducer from the first two steps to complete the proof of the Deatomization Theorem.

\input{semi-honest}

\input{to-honest}

\input{from-honest}

%% file: semi-honest.tex
\subsubsection{First step: a normal form}
\label{sec:to-honest-first}

In the first step of the proof of the Deatomization Theorem, we show that one can transform every pebble transducer for the deatomization into a certain normal form. In the normal form, all atom runs use configurations of maximal stack height $k$, in particular, an atom run does not use any push/pop operations on the pebble stack and can only modify the pebble stack by moving the head. Furthermore, when producing output, the head will only visit a constant number of atom blocks in the input.  Recall that the \emph{unit letter} is the letter $a$ used for the content of atom blocks $\langle a^n \rangle$. In the following lemma,  an \emph{opening configuration} is any configuration that is the first configuration  in some atom run. We assume without loss of generality that each opening configuration outputs exactly one opening bracket, and therefore for every input string, the opening configurations are in bijective correspondence with the atom blocks in the output string.

\begin{lemma} \label{lemma:semi-honest}
 If the  deatomization $\red f$ is computed by a $k$-pebble transducer, then it is also computed by a $k$-pebble transducer such that for some constant $d \in \set{1,2,\ldots}$, every atom run satisfies all of the following conditions:
 \begin{enumerate}
 \item \label{it:have-same-height} All configurations in the atom run have  stack height $k$.
\item \label{it:clean} In the opening configuration, none of the pebbles is over a unit position.
 \item \label{it:head-over-unit} In the remaining configurations, except the closing configuration, the head is over a unit position. 
 \item \label{it:few-atom-blocks} The unit positions visited by the head are located in at most $d$ atom blocks.
 \end{enumerate}
\end{lemma}
\begin{proof}
Consider a $k$-pebble transducer $\red f$ that computes the  deatomization. 
We begin by improving the transducer so that it satisfies a weakening of condition~\ref{it:have-same-height}: for every atom run, the  stack height of the opening configuration is minimal among the stack heights of the other configurations used in the same atom run. In other words, the topmost pebble from the opening configuration is not popped during the atom run. Later, we will ensure that the atom run also does not push pebbles, but this will require more work. 
 \begin{claim}\label{claim:opening-minimal}
 We can assume without loss of generality that in the $k$-pebble transducer, if the opening configuration in an atom run has stack height $\ell$, then all other configurations in this atom have stack height $\ge \ell$. 
 \end{claim}
 \begin{proof}
 Define the \emph{leading configuration} of an atom run to be the first configuration in the atom run among those that have minimal stack height. In this claim, we want to ensure that the leading configuration is the opening configuration. To see this, consider the set of pairs 
 \begin{align*}
 (\text{opening configuration of $\rho$}, \text{leading configuration of $\rho$}),
 \end{align*}
 where $\rho$ ranges over atom runs.
 This set of pairs is a partial bijection between configurations, which is definable in \mso. Therefore, we can create a new pebble transducer, which uses this bijection to swap the two configurations (and their outputs) in each run.
 \end{proof}

We now  improve the transducer so that it satisfies condition~\ref{it:clean}, i.e.~for every opening configuration all pebbles are over non-unit positions, i.e.~positions whose label is not the unit letter. The main observation is that the  original transducer  must already satisfy a certain weakening of this condition, as stated in the following claim. 
\begin{claim}\label{claim:radius}
There is a radius $r \in \set{0,1,\ldots}$ such that for every opening configuration,  each pebble is at most $r$ positions away from a non-unit position.
\end{claim}
\begin{proof}
Toward a contradiction, suppose that there is no such radius.   For every input string, the number of opening configurations in that string is the number of atom blocks in the  output string. The  set of opening configurations is definable in \mso.  If the claim would fail, then we could find opening configurations where  some pebble that is sufficiently far away from the nearest non-unit letter to apply a pumping argument with respect to the \mso formula defining the set of opening configurations.
By pumping a block of unit letters next to this position, we could create a different input string, in which there would be more opening configurations. Since this pumping would involve only unit letters, we would end up having two strings that differ only by the lengths of their atom blocks, but which have different numbers of atom blocks in the output. This cannot happen for the  deatomization. 
\end{proof}

Using the above claim, we can further improve the pebble transducer so that in every opening configuration, each pebble is over a non-unit position, as required by condition~\ref{it:clean} in the lemma. This is done by storing in the state the distance of each pebble to the nearest non-unit position; the numbers stored are taken from a finite set by the above claim. Using the same proof, we can also ensure a slightly stronger property: in every opening configuration,  the only part of an atom block where the pebbles are allowed is the opening bracket (i.e.~closing brackets are also disallowed).

Having ensured condition~\ref{it:clean} and a weaker version of condition~\ref{it:have-same-height}, we now move to ensuring the full version of~\ref{it:have-same-height}, as well as conditions~\ref{it:head-over-unit} and~\ref{it:few-atom-blocks} in the lemma. This will follow from a detailed analysis of the reachability relation, as presented in the following claim.

\begin{claim}
 \label{claim:spans-in-atom-runs} Let $\red f$ be a $k$-pebble transducer that computes the  deatomization, and let 
 $p$ and $q$ be states. Consider the property 
 \begin{align*}
 \varphi(\myunderbrace{x_1,\ldots,x_{\text{dim}(p)}}{$\bar x$}, \myunderbrace{y_1,\ldots,y_{\text{dim}(q)}}{$\bar y$})
 \end{align*}
 which holds in an input string $\red w$ if $p(\red w, \bar x)$ is the opening configuration in an atom run, and $q(\red w, \bar y)$ is some other configuration in the same atom run which outputs a unit letter. Then $\varphi$ is equivalent to a finite disjunction 
\begin{align*}
 \bigvee_{i \in I} \varphi_i(\bar x, \bar y),
\end{align*}
which is disjoint (at most one of the disjuncts holds for every input string $\red w$ and positions $\bar x \bar y$) and where each disjunct $\varphi_i$ has one of the following properties:
\begin{enumerate}
 \item {\bf Constant.} The variables $\bar x$ span $\varphi_i$ in the same sense as in the Basis Lemma, i.e.~for every input string $\red w$ one cannot find two different tuples that satisfy $\varphi_i$ and agree on the variables from $\bar x$; or 
 \item {\bf Linear.} There is a variable $y$ among $\bar y$ such that $\bar x y$ spans $\varphi_i$ in the sense described above. Furthermore, for every string $\red w$ with distinguished positions $\bar x$, there is a single atom block in $\red w$ which contains all positions $y$ that can be extended to a tuple $\bar y$ satisfying $\varphi_i(\bar x \bar y)$.
\end{enumerate} 
\end{claim}
\begin{proof}[Proof of Claim~\ref{claim:spans-in-atom-runs}]
We use the quantifier elimination result from Theorem~\ref{thm:quantifier-elimination}.    Let $\Delta$ be the input alphabet of the 
pebble transducer $\red f$. This alphabet is the disjoint union of $\Sigma$ with the two brackets and the unit letter used for representing atom blocks. By Theorem~\ref{thm:quantifier-elimination}, there is a tree grammar $\Tt$ and a linear interpretation $h$ such that the diagram
\[
 \begin{tikzcd}
 [column sep=1.5cm]
 \Delta^* 
 \ar[r,"\text{linear }h"]
 \ar[dr,equals, bend right=30]
 & \Tt 
 \ar[d, "\text{yield}"']
 \\
 & 
 \Delta^*,
 \end{tikzcd}
 \] 
 commutes and furthermore
 the \mso formula $\varphi(\bar x, \bar y)$ from the assumption of Claim~\ref{claim:spans-in-atom-runs} is equivalent to a quantifier-free formula $\psi( \bar x, \bar y)$ that works in the tree which is produced by $h$. Decompose the quantifier-free formula into a finite disjunction of quantifier-free types 
\begin{align*}
\bigvee_{i \in I} \psi_i(\bar x, \bar y).
\end{align*}
Recall the skeletons of quantifier-free types that were discussed in the proof of the Basis Lemma, and the corresponding notion of minimal \scc's. 
We will show that for every $i \in I$, the corresponding skeleton and its minimal \scc's satisfy the following condition:
\begin{itemize}
 \item[(*)] There is at most one minimal \scc that contains no variables from $\bar x$. Furthermore, if there is such a minimal \scc, then all leaves in that minimal \scc are labelled by the unit letter.
\end{itemize}

Before proving (*), we show how it implies the claim. The formula $\varphi_i(\bar x, \bar y)$ in the statement of the claim says that, after producing the tree computed by the linear interpretation $h$, the resulting leaves in the tree satisfy the quantifier-free type $\psi_i(\bar x, \bar y)$. Since every variable of $\psi_i$ is spanned by some variable in a minimal \scc, we conclude from (*) that all variables are spanned either by $\bar x$, or by $\bar x$ extended with one variable from $\bar y$. In the latter case, by the ``Furthermore'' part of (*), all values for that variable $y$ come from a single atom block, as required by the ``Furthermore'' part of the claim. 

It remains to prove (*). We use a pumping argument as in the proof of the Basis Lemma. Fix some $i \in I$, and suppose that the skeleton of $\psi_i$ has $c \in \set{0,1,\ldots}$ minimal \scc's that contain no variables from $\bar x$. We first show that the number $c$ of minimal \scc's is at most one, thus proving the first part of (*). Using the same argument for growth rates as in the Basis Lemma, we can use the  minimal \scc's without variables from $\bar x$ to create for each $n$ a tree in the tree grammar that has  $\Oo(n)$ leaves and  such that some tuple $\bar x$ of leaves in this tree can be extended in at least $n^c$ ways by variables $\bar y$ so that the result satisfies $\psi_i(\bar x \bar y)$. If $c \ge 2$, then the output of the corresponding atom block would be at least quadratic, and therefore it would exceed the size of every atom block in the input, leading to an atom block in the output that does not appear in the input. Therefore, $c \le 1$. The ``furthermore'' part of (*) is proved in the same way: if there would be a minimal \scc without variables from $\bar x$ but with a leaf labeled by a non-unit letter, then for each $n$ we could find a tree in the tree grammar with  $\Oo(n)$ leaves where all atom blocks in the corresponding input string have constant length, but the atom block produced by  some atom run has length at least $n$. This completes the proof of (*), and therefore also of Claim~\ref{claim:spans-in-atom-runs}.
\end{proof}

We now use the above claim to finish the proof of the lemma. Let $p$ be some state, of stack height $\ell \in \set{0,\ldots,k}$. By Claim~\ref{claim:opening-minimal}, we know that if an atom run begins with state $p$, then all configurations in this atom run use states of stack height  $\ge \ell$.
Fix $p$ and apply the claim for every state $q$, yielding a disjunction 
\begin{align*}
\bigvee_{i \in I_q} \varphi_i(\bar x, \bar y).
\end{align*}
Assume without loss of generality that all sets $I_q$ are disjoint, and let $I$ be their union.
If we assume that each configuration produces at most one output letter (which is easily assured for a pebble transducer), then for every atom run that begins with state $p$ and pebble stack $\bar x$, the size of the output for this run is 
\begin{align}\label{eq:output-as-sum}
\sum_{i \in I} \text{number of tuples $\bar y$ that satisfy $\varphi_i(\bar x, \bar y$)}.
\end{align}
The new pebble transducer will produce this output as follows. First, it pushes enough pebbles to make the stack have maximal size $k$; the stack will remain at this size throughout the run, thus ensuring condition~\ref{it:have-same-height} in the lemma. Recall that condition~\ref{it:clean} has already been assured using Claim~\ref{claim:radius}. Next, the transducer performs the following actions for every $i \in I$, depending on whether $i$ has constant or linear kind from Claim~\ref{claim:spans-in-atom-runs}:
\begin{itemize}
\item \textbf{Constant.} If $\varphi_i$ is spanned by $\bar x$, then  the number of tuples $\bar y$ that satisfy $\varphi_i(\bar x, \bar y)$ is zero or one, depending on an \mso definable property of $\bar x$. Therefore, the corresponding output can be produced already in the opening configuration of the atom run. 

\item \textbf{Linear.} Suppose now that $\varphi_i$ is spanned by some variable $y$ among $\bar y$. To produce the output corresponding to $\varphi_i$, we need to output one unit letter for each position $y$ that can be extended to some $\bar y$ that satisfies $\varphi_i(\bar x, \bar y)$. The transducer does this by looping through all such positions $y$. All choices for $y$ will use unit positions in a single atom block, thanks to the ``Furthermore...'' condition in Claim~\ref{claim:spans-in-atom-runs}; this will ensure items~\ref{it:head-over-unit} and~\ref{it:few-atom-blocks} in the lemma, with the constant $d$ in item~\ref{it:few-atom-blocks} being the size of $I$.

In order to loop through all positions $y$, the transducer needs to move the head to these positions. This raises the following issue: once the transducer has looped through all positions $y$, how does it recover the original placement of the head that was  used at the beginning of the atom run? If stack size in the original atom run was  non-maximal, i.e.~$\ell < k$, then this is not an issue, because we still have the original stack $\bar x$ on the first $\ell$ positions. However, if the original stack height in the first configuration of the atom run was $\ell = k$, then we seemingly run a risk of forgetting the topmost pebble in the original stack. However, in the case of $\ell = k$ the variable $y$ is necessarily the last variable in $\bar y$, since the first $k-1$ pebbles are not changed throughout an atom run thanks to Claim~\ref{claim:opening-minimal}. Therefore, in this case, the loop is simply running through an \mso definable subset of the original configurations in the atom run, and we can recover the original pebble stack since one can always recover in \mso the most recent configuration that produced an opening bracket.
\end{itemize}
\end{proof}

%% file: to-honest.tex
\subsubsection{Second step: some linear algebra}
\label{sec:to-honest-second}

By condition~\ref{it:few-atom-blocks} in Lemma~\ref{lemma:semi-honest},  there is some bound $d \in \set{0,1,\ldots}$ such that for every atom run, at most $d$ atom blocks from the input string are visited by the head. These atom blocks, i.e.~the atom blocks that are visited by the head, are called the \emph{significant blocks} of the atom run. Here is a picture of an atom run, represented by its first configuration, which has three significant blocks (in the run, $k=5$):
\mypic{24}
Note that in the picture, the positions of pebbles other than pebble $k$ are fixed through the atom run, since condition~\ref{it:have-same-height} says that all configurations in an atom run have maximal stack height $k$. The only pebble that moves during the run is the head, i.e.~pebble $k$; the possible positions for this pebble are depicted using blue arrows in the picture.
As in the picture, we number the significant blocks according to their position in the input string. Define the \emph{profile} of an atom run to be the vector in $\Nat^d$ which describes the lengths of its significant blocks, ordered according to their position in the input string. If the number of significant blocks is smaller than $d$, then this vector is padded with zeros. For example, if $d=5$ and we consider the atom run in the picture above, then its profile is (4, 5, 4, 0, 0), because there are three significant blocks with respective lengths 4, 5, 4, and the fourth and fifth significant block are undefined. 

To prove Lemma~\ref{lem:honest}, we will analyze in detail the relationship between the profile of an atom run and the length $n$ of the atom block $\langle a^n \rangle$ that is produced by the atom run.
The first step of this analysis is Lemma~\ref{lem:affine-partition}, which says that this relationship is described by an affine function that depends only \mso definable properties of the atom run.
%  When talking about an \mso definable property of atom runs in the lemma, we use the following notion:  an \mso definable partition of the atom runs is a partition of the atom runs into finitely many parts, such that for each part, the set of corresponding opening configurations is \mso definable. 

\begin{lemma}\label{lem:affine-partition}
 Consider a $k$-pebble transducer obtained by applying Lemma~\ref{lemma:semi-honest}. There is  partition of the opening configurations into finitely many \mso definable parts, and for each part $P$ of this partition there is a function 
 \begin{align*}
 (n_1,\ldots,n_d) \in \Nat^d \quad \mapsto \quad \myunderbrace{\lambda_0 + \lambda_1 n_1 + \cdots + \lambda_d n_d}{$\lambda_0$ is an integer and $\lambda_1,\ldots,\lambda_d$ are\\ non-negative rationals coefficients},
 \end{align*}
 such that for every atom run with opening configuration in part $P$, the length of its output is obtained by applying the function to the profile.
\end{lemma}
\begin{proof}
By conditions~\ref{it:head-over-unit} and~\ref{it:few-atom-blocks} from Lemma~\ref{lemma:semi-honest}, the length of the output produced by an atom run with $d$ significant blocks is equal to 
 \begin{align}
 \label{eq:sum-of-significant-blocks}
 \sum_{i \in \set{1,\ldots,c}} 
 \begin{tabular}{l}
    \text{number of configurations with the}\\
    \text{head in the $i$-th significant block.}
 \end{tabular}
 \end{align}
 Therefore, to prove the lemma, we will show that the $i$-th summand in the sum above can be expressed using an affine function applied to the length of the $i$-th significant block. This will be done using the following observation about the output sizes of \mso queries the select positions in inside atom blocks.
\begin{claim}\label{claim:affine-relationship}
 For every \mso formula  $\varphi(x)$  there is a partition of the set $a^*$  of all words over a one letter alphabet, such that the partition has finitely many \mso definable parts, and for each  part there are  associated coefficients $\lambda_0,\lambda$, such that for every word $a^n$ in this part, there are exactly $\lambda_0 + \lambda n$ positions selected by $\varphi(x)$.
\end{claim}
\begin{proof}
 By the equivalence of \mso and regular languages, there is some number $k$ such that whether or not $\varphi(x)$ selects a position $i$ in some word $a^n$ depends only on its \emph{type}, which is defined to be the  following four numbers:
 \begin{align*}
 i \shortmod k \qquad n-i \shortmod k \qquad \min(i,k) \qquad \min(n-i,k).
 \end{align*}
 There are finitely many types, and for every $n$, the number of positions with a given type in the word $a^n$ is equal to $\lambda_0 + \lambda n$, with the coefficients $\lambda_0$ and $\lambda$ depending on the type as well as the values of $\min(n,k)$  and $n \shortmod k$. Since the latter values can be described by regular languages, the claim follows.  Note that the coefficient $\lambda$ need not be an integer, e.g.~it is equal to $\frac 1 2$ when the property $\varphi(x)$ says that $x$ is an even-numbered position.
\end{proof}

\begin{corollary}\label{cor:lambdas}
 Let $\varphi(y)$ be an \mso formula, and let $i \in \set{1,\ldots,c}$. There is a  partition of the opening configurations, such that the partition has finitely many \mso definable parts, and for  each part there are associated   coefficients $\lambda, \lambda_0$, such that for every  opening configuration in the part, the number of positions selected by $\varphi(y)$ in the $i$-th significant block of the corresponding atom run is
 \begin{align*}
 \lambda_0 + \lambda\cdot\text{length of $i$-th significant block}.
 \end{align*}
\end{corollary}
\begin{proof}
 If we take an atom run, and we ask about whether or not an \mso formula $\varphi(y)$ holds in some position of the $i$-th significant block of this atom run, then the answer to the question will depend only on \mso definable properties of the following three parts of the opening configuration:
 \mypic{25}
 The only variable present in the middle part is $y$, since the middle part cannot contain any pebbles by condition~\ref{it:clean} of Lemma~\ref{lemma:semi-honest}.
 The appropriate \mso information about the left and right parts can be fixed by the \mso partition of the opening configurations; while to the middle part, we can apply the analysis from Claim~\ref{claim:affine-relationship}.
\end{proof}

The lemma follows  by applying the above corollary to each summand in~\eqref{eq:sum-of-significant-blocks}.
\end{proof}

 In the lemma above, we have shown that the output is determined by applying some affine function to the profile. In principle, the affine function could mix the significant blocks, e.g.~the output length could be the average of the lengths of the first two significant blocks. The following lemma rules out such mixing, by using a more refined analysis of the coefficients used in affine functions.

\begin{lemma}\label{lem:affine-partition-two}
 One can strengthen the conclusion of Lemma~\ref{lem:affine-partition} so that for every part, at most one of the coefficients $\lambda_0,\ldots,\lambda_d$ is nonzero, and furthermore if $\lambda_i$ is nonzero for $i \neq 0$, then $\lambda_i=1$.
\end{lemma}

\begin{proof}
    We begin the proof with an analysis of the possible profiles that can arise in an \mso definable set of atom runs. We show that sets of profiles arising this way can be assumed to be products of arithmetic progressions. Here,  an \emph{arithmetic progression} is  a subset of the natural numbers that is of the form $\alpha + \beta \Nat$ for some $\alpha,\beta \in\Nat$. Examples of arithmetic progressions include the odd numbers, or the singleton set $\set{7}$. A non-example is the set of numbers not divisible by three, although this set is a union of two arithmetic progressions.

 \begin{claim}
 One can refine the partition from Lemma~\ref{lem:affine-partition} so that for every part, the corresponding set of profiles is of the form 
 \begin{align*}
 \myunderbrace{\Pi_1 \times \cdots \times \Pi_d}{each $\Pi_i$ is an arithmetic progression}.
 \end{align*}
 \end{claim}
 \begin{proof}
If we view a regular (equivalently, \mso definable) language over a unary alphabet as a set of natural numbers, then this set of numbers will be ultimately periodic, i.e.~it will be a finite union of arithmetic progressions.
 By a compositionality argument similar to the one used in the proof of Corollary~\ref{cor:lambdas}, whether or not the first configuration of an atom run satisfies some \mso  formula depends only on the \mso definable properties of  the significant blocks, and the at most $d+1$ parts of the configuration that are separated by the significant blocks. Which formulas are true in the significant blocks depends only on their lengths, in an ultimately periodic way. By distributing union over products, we get a decomposition as in the statement of the claim. 
 \end{proof}

 From now on, we assume that the partition from Lemma~\ref{lem:affine-partition} has been refined according to the above claim. Take some  part $P$ of this partition, with the  corresponding set of profiles being
 \begin{align*}
 \Pi = \Pi_1 \times \cdots \times \Pi_d \subseteq \Nat^d.
 \end{align*}
From Lemma~\ref{lem:affine-partition}, we know that for every atom run with its opening configuration in $P$, the length of its output is obtained by applying to its profile some affine function 
\begin{align*}
    % \label{eq:an-affine-function}
f(n_1,\ldots,n_d) = \lambda_0 + \lambda_1 n_1 + \cdots + \lambda_d n_d
\end{align*}
which depends only on the part.
Without loss of generality, we can assume that: (*) if the arithmetic progression $\Pi_i$ is finite (i.e.~it describes a singleton set), then the coefficient $\lambda_i$ is zero. This is because the contribution of the $i$-th coordinate can be put into the constant $\lambda_0$ when the $i$-th coordinate is known to be fixed. 
To prove the lemma, we will show that, assuming (*), at most one of the coordinates $\lambda_0,\ldots,\lambda_d$ can be nonzero, and if the nonzero coordinate is not $\lambda_0$, then it must be equal to 1.

We first show that at most one of the coordinates $\lambda_1,\ldots,\lambda_d$ can be nonzero. This will follow from two observations. The first observation, see Claim~\ref{claim:affine-two-combine}, is that if two of the coefficients would be nonzero, then we could use them to combine coordinates in a non-trivial way. The second observation, see Claim~\ref{claim:pumping}, will show that the non-trivial combinations are forbidden.

\begin{claim}\label{claim:affine-two-combine}
 If at least two of the coordinates $\lambda_1,\ldots,\lambda_d$ are nonzero, then the following set is infinite
 \begin{align*}
 \setbuild{ f(p)}{$p \in \Pi$ and $f(p)$ is not equal to any coordinate in $p$}.
 \end{align*} 
\end{claim}
\begin{proof}
    Suppose that the arithmetic progression $\Pi_i$ is $\alpha_i + \beta_i n$.  If we take any affine function which has at least two nonzero coefficients $\lambda_1,\ldots,\lambda_d$, then for sufficiently large $n$ applying the function to  the vector
    \begin{align*}
     (\alpha_1 + \beta_1 n, \alpha_2 + \beta_2 n^2, \ldots, \alpha_d + \beta_d n^d) \in \Pi
    \end{align*}
 will yield an output that does not appear in any of the coordinates of the vector. This is because smaller coordinates in the vector will be to small to cancel out the contributions of the larger coordinates.
\end{proof}

The second observation is that every profile can be seen as arising from some atom run where all non-significant blocks have bounded length. 

\begin{claim}\label{claim:pumping}
 There is a constant $n_0$, such that every atom run in $P$ has the same profile as some atom run in $P$ where all non-significant blocks have length at most $n_0$.
\end{claim}
\begin{proof}
 A pumping argument.
\end{proof}

Using the above two claims, we conclude that at most one of the coordinates $\lambda_1,\ldots,\lambda_d$ can be nonzero. Indeed, otherwise, by Claim~\ref{claim:affine-two-combine} we could find some profile $p \in \Pi$ such that $f(p)$ is bigger than the constant $n_0$ from Claim~\ref{claim:pumping}; this would lead to an atom run whose output atom block is not equal to any atom block from the input, which cannot happen for the de-atomization. 

So far, we have proved that at most one of the coefficients $\lambda_1,\ldots,\lambda_d$ is nonzero. To finish the proof of the lemma, we will show that if $\lambda_i$ is nonzero for  $i \in \set{1,\ldots,d}$, then $\lambda_i=1$ and $\lambda_0=0$. To see this, consider some profile in $\Pi$ where the $i$-th coordinate is much bigger than all the other coordinates. Such a profile exists, since the $i$-th arithmetic progression $\Pi_i$ is infinite by assumption (*). By Claim~\ref{claim:pumping}, this profile arises from an atom run where the only large atom block in the input is the $i$-th significant block; since $\lambda_i$ is nonzero it follows that the output block is too large to be equal to anything but the $i$-th significant block, and therefore $\lambda_i=0$ and $\lambda_0=0$. 
\end{proof}

%% file: from-honest.tex
\subsubsection{Third step: putting it all together}
\label{sec:final-proof-of-deatomization}
In this section, we complete the proof of the Deatomization Theorem. We begin with the following lemma, which combines the results of the analysis from  Sections~\ref{sec:to-honest-first} and~\ref{sec:to-honest-second}.

\begin{lemma}\label{lem:honest}
    If the deatomization $\red f$ is computed by some $k$-pebble transducer, then it is computed by one which has the following properties: 
    \begin{enumerate}
        \item \label{it:super-clean} in every opening configuration, none of the pebbles is over a unit position, or over a closing bracket of an atom block;
        \item \label{it:big-or-small} there is a constant $n_0 \in \set{0,1,\ldots}$ such that for every atom run, either:
       \begin{enumerate}
        \item the output of the atom run has length at most $n_0$; or 
        \item the head in the opening configuration is over the opening bracket of some input atom block $\langle a^n \rangle$, and the output of the atom run is equal to $\langle a^n \rangle$.
       \end{enumerate}
        \end{enumerate}
\end{lemma}
\begin{proof}
    Apply Lemma~\ref{lemma:semi-honest} to the pebble transducer from the assumption, and fix the resulting pebble transducer for the rest of this proof. As we have remarked after the proof of Claim~\ref{claim:radius}, the  transducer  already satisfies condition~\ref{it:super-clean} from the present lemma. We will now improve it so that it also satisfies condition~\ref{it:big-or-small}. 

    Consider the \mso definable partition of opening configurations from Lemmas~\ref{lem:affine-partition} and~\ref{lem:affine-partition}. In this partition, every part has one of two properties: either (a) all atom runs corresponding to this part have the same output; or (b) there is some $i \in \set{1,\ldots,d}$ such that all atom runs corresponding to this part copy the $i$-th significant block from the input to the output. The outputs of atom runs of kind (a) have bounded length, since there are finitely many parts; let $n_0$ be the maximal length that arises this way. To finish the proof of the lemma, we will modify the transducer so that for each atom run of kind (b), the first configuration has its head over the $i$-th significant block. For an atom run of kind (b), consider the first configuration in this atom run with the head in the $i$-th significant block. We can use the same  swapping argument as in the proof of Claim~\ref{claim:radius} to ensure that this is the opening configuration. Finally, if the opening configuration does not have its head over the opening bracket, and instead uses some unit position inside the atom block, then we can precede it by a configuration just before which does have its head over the opening bracket. 
\end{proof}

The pebble transducer in the above lemma is almost the same as a pebble transducer with atoms, if we ignore the contents of atom runs and simply think of them as outputting atoms in a single step. The only extra tricks that our deatomized pebble transducer can play, and which would not be available to pebble transducer with atoms, are: (a) outputting atom blocks of constant size without having the head over such atom blocks in the input string; and (b) testing regular properties of strings that represent atoms, such as ``even length''. In the final part of the proof of the deatomization Theorem, we show that such tricks are useless if the atom representation is chosen so that: (a) short atom blocks do not appear in the input; and (b) all atom blocks have the same regular properties.

    Consider a  $k$-pebble transducer which computes the deatomization $\red f$, and assue that it satisfies the conditions from  Lemma~\ref{lem:honest}.  Let $r \in \set{1,2,\ldots}$ be the maximal quantier rank of \mso formulas defining reachability on the configurations of this pebble transducer. By the pigeon-hole principle, we can choose an atom representation 
    \begin{align*}
     \alpha : \atoms \to  \langle a^* \rangle,
    \end{align*}
    which is injective, where all atom blocks used to represent atoms are longer than the constant $n_0$ from Lemma~\ref{lem:honest} and also have the same \mso theory of quantifier rank $r$. Because all atom blocks are longer than $n_0$,  the output of of every atom run is equal to the atom block that is pointed to by the head in its first configuration thanks to condition~\ref{it:big-or-small} in Lemma~\ref{lem:honest}.

    For an input  string $w$ with atoms, we can injectively map the positions of $w$ to the positions of its deatomization $\alpha(w)$, by mapping atoms to the opening brackets of the corresponding atom blocks, as explained in the following picture 
    \mypic{23}
     If $\bar x$ is a tuple of distinguished positions in $w$, then  we write $\alpha(w,\bar x)$ for the string $\alpha(w)$ together with the distinguished positions that correspond to $\bar x$ under the injective map described above. 
    
    \begin{claim}\label{claim:pull-back-across-alpha} Let $w, \bar x$ be a string with atoms together with distinguished positions. The \mso theory of rank $r$  of $\alpha(w, \bar x)$ is  uniquely determined by the \mso theory of rank $r$ of $w,\bar x$.
    \end{claim}
    \begin{proof}
        Because all code blocks in $\alpha(w)$ have the same \mso theory of quantifier rank $r$.
    \end{proof}
    
    Using the above claim, we  define a new pebble transducer with atoms as required by the conclusion of the Deatomization Theorem. The states and their stack heights are the same as in the transducer from the assumption of the theorem.  The transitions in the new pebble transducer (with atoms) are designed so that for every input string $w$ (with atoms), if we apply $\alpha$ to the accepting run, then the result is the accepting run of  the original pebble transducer (without atoms) over the input string $\alpha(w)$; this can be done thanks to Claim~\ref{claim:pull-back-across-alpha}. The output instructions in the new pebble transducer are taken from the  original one.